\documentclass[journal]{IEEEtran}
\usepackage{amssymb}
\usepackage{mathrsfs}
\usepackage{graphicx}
\usepackage{amsmath, amssymb, amsthm}
\usepackage{leftidx}
\usepackage{extarrows}
\usepackage{stfloats}
\usepackage{picinpar}
\usepackage{enumerate}
\usepackage{algorithm}
\usepackage{algorithmicx}
\usepackage{algpseudocode}
\usepackage{epsfig}
\usepackage{latexsym}
\usepackage{amsfonts}
\usepackage{enumerate}
\usepackage{graphics}
\usepackage{graphicx,subfigure}
\usepackage{MnSymbol}
\usepackage{float}
\usepackage{pict2e}
\usepackage{tikz}
\usepackage{bm}
\newtheorem{theorem}{Theorem}

\newtheorem{corollary}{Corollary}
\newtheorem{lemma}{Lemma}
\newtheorem{remark}{Remark}
\newtheorem{definition}{Definition}

\newtheorem{problem}{Problem}
\newtheorem{property}{Property}

\title{A General Control Framework for Boolean Networks}

\author{Shiyong~Zhu,~\IEEEmembership{Student Member,~IEEE}, Jianquan Lu$^\ast$,~\IEEEmembership{Senior Member,~IEEE}, Shun-ichi Azuma,~\IEEEmembership{Senior Member,~IEEE}, and Wei Xing Zheng,~\IEEEmembership{Fellow,~IEEE}
\thanks{This work was supported by the National Natural Science Foundation of China under Grant No. 61973078, and ``333 Engineering" Foundation of Jiangsu Province of China under Grant BRA2019260.}
\thanks{$^\dag$Corresponding author: Jianquan Lu.}
\thanks{Shiyong Zhu and Jianquan Lu are with the Department of Systems Science, School of Mathematics, Southeast University, Nanjing 210096, China (email: zhusy0904@gmail.com; jqluma@seu.edu.cn).}
\thanks{Shun-ichi Azuma is with the Graduate School of Engineering, Nagoya University, Furo-cho, Chikusa-ku, Nagoya, 464-8603, Japan (e-mail: shunichi.azuma@mae.nagoya-u.ac.jp).}
\thanks{Wei Xing Zheng is with the School of Computer, Data and Mathematical Sciences, Western Sydney University, Sydney, NSW 2751, Australia (e-mail: w.zheng@westernsydney.edu.au).}
}

\begin{document}
\maketitle
\thispagestyle{empty}
\pagestyle{empty}
\begin{abstract}
  This paper focuses on proposing a general control framework for large-scale Boolean networks (\texttt{BNs}). To release the dependency of node dynamics in traditional framework, the concept of structural controllability for \texttt{BNs} is formalized. A necessary and sufficient criterion is derived for the structural controllability of \texttt{BNs}; it can be verified with $\Theta(n^2)$ time, where $n$ is the number of network nodes. An interesting conclusion is shown as that a \texttt{BN} is structurally controllable if and only if it is structurally fixed-time controllable. Afterwards, the minimum node control problem with respect to structural controllability is proved to be NP-hard for structural \texttt{BNs}. In virtue of the structurally controllable criterion, three difficult control issues can be efficiently addressed and accompanied with some advantages. An open problem-network aggregation for the controllability of \texttt{BNs}-is addressed by a computationally efficient aggregation strategy, which provides an approach to control the minimal number of nodes. In terms of the design of pinning controllers to generate a controllable \texttt{BN}, by utilizing the structurally controllable criterion, the selection procedure for the pinning node set is developed for the first time instead of just checking the controllability under the given pinning control form; the pinning controller is of distributed form, and the time complexity is $\Theta(n2^{3d^{\ast}}+2(n+m)^2)$, where $m$ and $d^\ast$ are respectively the number of generators and the maximum vertex in-degree. With regard to the control design for stabilization in probability of probabilistic \texttt{BNs} (\texttt{PBNs}), an important theorem is proved to reveal the equivalence between several types of stability. The existing difficulties on the stabilization in probability are then solved to some extent via the structurally controllable criterion. Finally, numerical examples are employed to demonstrate several applications of theoretical results.
\end{abstract}


\section{Introduction}
Boolean network (\texttt{BN}) is a classical type of discrete-time dynamical models with binary state variables. The research history of \texttt{BNs} traces back to the model of gene regulatory networks, which was originally proposed by Kauffman \cite{kauffman1969jtb437}. Since \texttt{BNs} can describe the qualitative behaviors of many engineering mechanisms, this type of models has been extensively applied to many areas including systems biology \cite{davidson2002science1669,gaozg2017tcns770}, game theory \cite{chengdz2018auto51}, multi-agent systems \cite{fagiolini2013auto2339} and so on. It has prompted scholars to start intensive research with using such canonical discrete-time and two-valued dynamical systems. Let state variables ${\bm x}_i$, $i \in \{1,2,\cdots,n\}$, and control inputs ${\bm u}_j$, $j \in \{1,2,\cdots,m\}$, be either 1 or 0. Then, \texttt{BN} (\ref{e-LN}) and Boolean control network (\texttt{BCN}) (\ref{e-LCN}) can be respectively given as
\begin{equation}\label{e-LN}
\begin{aligned}
{\bm x}_k(t+1)&={\bm f}_k([{\bm x}_i(t)]_{i\in {\bf X}_k}),~k\in\{1,2,\cdots,n\},\\
\end{aligned}
\end{equation}
\begin{equation}\label{e-LCN}
\begin{aligned}
{\bm x}_k(t+1)&=\tilde{{\bm f}}_k([{\bm x}_i(t),{\bm u}_j(t)]_{i\in \tilde{{\bf X}}_k,j\in \tilde{{\bf U}}_k}),~k\in\{1,2,\cdots,n\},
\end{aligned}
\end{equation}
where the minimally represented functions ${\bm f}_k$ and $\tilde{{\bm f}}_k$ respectively capture the dynamics of the $k$-th node of \texttt{BN} (\ref{e-LN}) and \texttt{BCN} (\ref{e-LCN}). While the set
${\bf X}_k$ denotes the index set of functional state variables of ${\bm f}_k$, the sets $\tilde{{\bf X}}_k$ (respectively, $\tilde{{\bf U}}_k$) collects the indices of the functional state variables (respectively, input variables) of $\tilde{{\bm f}}_k$.

Another dramatic breakthrough for the study of \texttt{BNs} is the proposal of algebraic state space representation (\texttt{ASSR}) for \texttt{BNs}, based on the semi-tensor product (\texttt{STP}) of matrices ``$\ltimes$''
\cite{chengdz2011springer}. In this setting, the \texttt{ASSR} of \texttt{BN} (\ref{e-LN}) and \texttt{BCN} (\ref{e-LCN})
are respectively established as
\begin{equation}\label{equ-assr-bn}
x(t+1)=Lx(t),
\end{equation}
\begin{equation}\label{equ-assr-bcn}
x(t+1)=\tilde{L}u(t)x(t),
\end{equation}
where $x(t):=\ltimes_{i=1}^{n}x_i(t)$ and $u(t):=\ltimes_{j=1}^{m}u_j(t)$ with $x_i(t)$ and $u_j(t)$ respectively being the canonical vectors of ${\bm x}_i(t)$ and ${\bm u}_j(t)$.

After that, dynamics analysis and control design for \texttt{BNs} have been widely considered. To just list a few, controllability
\cite{chengdz2009tac1659,margaliot2012aut1218,zhaoy2010scl767,lujq2016ieeetac1658}, observability \cite{chengdz2009tac1659,Valcher2012TAC1390,zhangkz2020tac}, stabilization \cite{guoyq2015auto106,lir2013ieeetac1853}, optimal control \cite{valcher2013tac1258,wuyh2019auto378}, as well as the decoupling problems \cite{Yuyongyuan2019TAC,zhujd2020TAC}.

\subsection{Motivations and Related Works}
As a basic but crucial property for control systems, controllability plays an important role in the study of many control issues, such as the stabilization of unstable systems by feedback control and the optimal control
\cite{sontag2013mathematical}. In the area of \texttt{BNs}, the concept of controllability was defined by Akutsu {\em et al.} for the first time, and the verification for controllability was proved to be NP-hard \cite{akutsu2007jtb670}. In \cite{chengdz2009tac1659}, Cheng and Qi defined the input-state transfer matrix to necessarily and sufficiently characterize the controllability of \texttt{BCNs}; but the dimension of this transfer matrix exponentially increases with the growth of the vertex number, the input number as well as the evolution time. In order to break through this dimension-increasing limitation, two conditions were respectively established by Zhao {\em et al.} through the input-state incidence matrix \cite{zhaoy2010scl767} and by Laschov {\em et al.} via the Perron-Frobenius theory \cite{margaliot2012aut1218}. However, as mentioned in \cite{zhaoy2010scl767}, the \texttt{ASSR} approach was only suitable for checking the controllability of \texttt{BNs} with the vertex number being less than $25$ or so. Even if several attempts to reduce the complexity of controllability were made in succession \cite{liangjl2017tac6012,zhuqx2018tac}, the time complexity still remained at $\Theta(2^{2n})$ at least.

On the one hand, the size of matrix $L$ (respectively, $\tilde{L}$) in \texttt{BN} (\ref{equ-assr-bn}) (respectively, \texttt{BCN} (\ref{equ-assr-bcn})) is $2^n\times 2^n$ (respectively, $2^n\times2^{n+m}$). Since many realistic systems always possess a large mass of nodes, ({\bf P1}) {\bf such exponentially-increasing complexity with respect to (w.r.t.) the growth of nodes is unacceptable for large-scale \texttt{BNs}}. On the other hand, supported by experiments, Azuma {\em et al.} formulated that the identification of network structure is easier to access than that of node dynamics \cite{Azuma2019TCNS464,Azuma2015TCNS179}. Therefore, in some practical senses, the full information of node dynamics cannot be obtained. However, ({\bf P2}) {\bf the minute information of node dynamics is required in all the traditional results on the controllability analysis based on the \texttt{STP} of matrices.} It is thus very necessary to develop an efficient framework for the structural controllability of \texttt{BCNs}.

For a standard linear time-invariant (\texttt{LTI}) system
\begin{equation}\label{equ-LTI}
\frac{\mathrm{d}{\bf x}(t)}{\mathrm{d}t}=A{\bf x}(t)+B{\bf u}(t),
\end{equation}
where vector ${\bf x}(t)=[x_1(t),x_2(t),\cdots,x_n(t)]^\top$ represents the state of a system of $n$ nodes at time $t$, its structural controllability has been well defined
in \cite{linct1974tac-structural} as that \texttt{LTI} system (\ref{equ-LTI}) is called structurally controllable if, there exists a completely controllable pair $(A_0,B_0)$, which processes the same structure as $(A,B)$, i.e.,
$(A_0,B_0)\sim (A,B)$. This definition is partly motivated by the proposition that one can always find a completely controllable pair $(A_0,B_0)$ satisfying $(A_0,B_0)\sim (A,B)$, $\parallel A'-A \parallel\leq \varepsilon$ and
$\parallel B'-B \parallel\leq \varepsilon$ for any $\varepsilon>0$. Supported by this proposition, the inaccuracy of parameters in (\ref{equ-LTI}) can be ignored to some extent. However, this motivation is not true for \texttt{BCNs}. Thus, we need to define the structural controllability of \texttt{BNs} in a similar manner as strongly structural controllability
\cite{mayeda1979siam-strongcontrollability} by resorting of the structural equivalence of \texttt{BCNs} \cite{Azuma2019TCNS464,Azuma2015TCNS179}.

Of course, it leads to another meaningful problem: if \texttt{BCN} (\ref{e-LCN}) is not-structurally-controllable, how to control some nodes so as to make this network structurally controllable? As the injection of control nodes is costly and timely, the minimization of the number of control inputs is of theoretical and practical significance. Thus, we would like to study the following problem.
\begin{problem}\label{problem-minimalnodecontrol}
Given \texttt{BN} (\ref{e-LN}), determine the minimal set $\Lambda^\ast \subseteq \{1,2,\cdots,n\}$ such that the following \texttt{BCN} is structurally controllable:
\begin{equation}\label{equation-LCN}
\left\{\begin{aligned}
&{\bm x}_k(t+1)={\bm f}_k([{\bm x}_i(t)]_{i\in {\bf X}_{k}}), k\not\in \Lambda^\ast,\\
&{\bm x}_k(t+1)={\bm u}_k(t), k\in \Lambda^\ast.
\end{aligned}
\right.
\end{equation}
\end{problem}
A similar attempt for \texttt{LTI} system (\ref{equ-LTI}) was made in \cite{liuyy2011nature167}, and the answer is that the minimal number of inputs is equal to the number of unmatched nodes while maximally matching the graph w.r.t. pair $(A,B)$. Therefore, we also solve the question: ({\bf P3}) {\bf find the effective strategy to solve Problem \ref{problem-minimalnodecontrol}.} Until now, some efforts on problems ({\bf P2}) and Problem \ref{problem-minimalnodecontrol} have been made. For instance, in \cite{Margaliot2019TAC2727}, the controllability of conjunctive \texttt{BNs} (\texttt{CBNs}) was verified in polynomial time; and Problem \ref{problem-minimalnodecontrol} was proved to be NP-hard for \texttt{CBNs}. Nonetheless, logical couplings in \texttt{CBNs} are unique; thus Problem \ref{problem-minimalnodecontrol} is more complex for general \texttt{BCNs}. Besides, an efficient search algorithm for problem ({\bf P3}) is still lacking even if for \texttt{CBNs}.

In the field of \texttt{BNs}, one biology-inspired control strategy is pinning control; it is motivated by the phenomenon that the overall behavior of a large biological system can be derived by controlling a fraction of
nodes. For instance, one can provoke the entire body of a worm via $49$ neurons on average (approximately 16.5\%) \cite{cho2011scientific}. The concept of pinning controllability for \texttt{BNs} was firstly proposed by Lu {\em et al.} \cite{lujq2016ieeetac1658} by injecting the external inputs only on a small fraction of network nodes:
\begin{equation}\label{e-pinning-LN}
\left\{\begin{aligned}
{\bm x}_k(t+1)&={\bm u}_k \oplus_k {\bm f}_k([{\bm x}_i(t)]_{i\in {\bf X}_k}),~k\in \Gamma,\\
{\bm x}_k(t+1)&={\bm f}_k([{\bm x}_i(t)]_{i\in {\bf X}_k}),~k\not\in \Gamma,
\end{aligned}\right.
\end{equation}
where $\Gamma$ is the index set of pinning nodes, ${\bm u}_k$ is the external control input, and $\oplus_j$ is the binary logical operator.

In this setup, the pinning controllability of autonomous \texttt{BNs} was further investigated in \cite{chenhw2016scis}. Additionally, the pinning controllability of \texttt{BNs} was also studied via injection modes \cite{liuzq2020tcns}. However, except for ({\bf P1}), another limitation still exists in this topic, as that ({\bf P4}) {\bf the pinning nodes must be assigned in advance}. That is, all the results in \cite{lujq2016ieeetac1658,chenhw2016scis,liuzq2020tcns} only provide a procedure to check whether the pre-assigned pinning control can make \texttt{BCNs} (\ref{equation-pinning-bn}) controllable rather than to develop a design procedure. Although pinning control is pretty efficient, particularly for large-scale \texttt{BNs}, problems ({\bf P1}) and ({\bf P4}) remain its drawback. Again, it should be noted that ``({\bf P5}) {\bf finding a computationally efficient sufficient condition on the aggregated subnetworks, for the purpose of checking controllability of the entire \texttt{BCN} (\ref{e-LCN}),} is still an open problem \cite{zhaoy2015tnnls}."

Furthermore, since the intrinsic fluctuations and extrinsic perturbations typically affect the regulatory interactions among DNA, RNA and proteins, the model of probabilistic \texttt{BNs} (\texttt{PBNs}) was formalized by Shmulevich {\em et al.} \cite{Shmulevich2002PBN} as
\begin{equation}\label{equ-pbn}
{\bm x}_k(t+1)={\bm f}^{\sigma(t)}_k([{\bm x}_i(t)]_{i\in {\bf X}^{\sigma(t)}_k}),~k\in\{1,2,\cdots,n\},
\end{equation}
where stochastic sequence $\sigma(t)\in\{1,2,\cdots,s\}$ is the independent and identically distributed process w.r.t. probability vector ${\bm p}=[p^1,p^2,\cdots,p^{s}]$. Once the switching mode $\sigma(t)$ is given, the meanings of other notations in (\ref{equ-pbn}) are consistent with those in (\ref{e-LN}). In \texttt{PBNs}, stabilizability in probability, which was recently proposed in \cite{huangc2020ins205} and \cite{huangc2020tnnls}, is of general concern. Due to the lack of recursion of reachable sets, it results in that problems ({\bf P1}) and ({\bf P4}) also exist in this topic. Thus, it is pretty difficult to design an efficient control strategy.

\subsection{Contributions}
In this paper, we establish a novel and general framework for the control of \texttt{BNs} to overcome the above problems ({\bf P1})-({\bf P6}). The main contributions of this paper are listed as follows:
\begin{itemize}
  \item[I.] By resorting to the structural equivalence, structural controllability is defined as a novel concept in the field of \texttt{BNs} to overcome the problems ({\bf P1}) and ({\bf P2}). A polynomial-time necessary
  and sufficient condition is developed to judge the structural controllability based on some concepts firstly proposed in \cite{margaliot2018auto56}. Additionally, Problem \ref{problem-minimalnodecontrol} is proved to
  be NP-hard.
  \item[II.] In terms of open problem ({\bf P5}), a feasible network aggregation strategy is derived for controllability analysis of large-scale \texttt{BNs} on the basis of structural controllability theorem, to fill up the gap on this topic. We then use the network aggregation method and the \texttt{STP} approach to develop a computational algorithm for problem ({\bf P3}), which is still the lack for \texttt{CBNs} \cite{margaliot2018auto56}.
  \item[III.] Via the structural controllability criterion, the distributed pinning controller is developed to make an arbitrary \texttt{BCN} controllable. By comparison to the traditional methods in
  \cite{lujq2016ieeetac1658,chenhw2016scis,liuzq2020tcns,liff2020tcns1523}, the time complexity is reduced from $\Theta(2^{2n})$ to $\Theta(n2^{3d^\ast}+2(n+m)^2)$, where $n$, $m$ and $d^\ast$ are respectively the node number, the input number and the largest vertex in-degree. In practice, biological networks are sparsely connected \cite{jeong2000nature,jeong2001nature}, thus the largest vertex in-degree $d^\ast$ would not be pretty large and our approach is able to handle the problem ({\bf P1}) in \cite{lujq2016ieeetac1658,chenhw2016scis,liuzq2020tcns} to some extent. Compared with \cite{lujq2016ieeetac1658,chenhw2016scis,liuzq2020tcns}, the pinning nodes are designable by a polynomial-time algorithm, so the problem ({\bf P4}) is well overcome.
  \item[IV.] With regard to question ({\bf P6}) for the stabilization in probability of \texttt{PBNs}, different from the results in \cite{huangc2020ins205} and \cite{huangc2020tnnls}, the pinning nodes here are
  also accordingly chosen rather than pre-assigned. Besides, an interesting theorem is proposed to verify the equivalence of different types of stability. Via the structurally controllable condition, the controller can be easily designed to make \texttt{PBN} (\ref{equ-pbn}) stable in probability.
\end{itemize}

\subsection{Organization}
The outline of this paper is as follows. Section \ref{sec-preliminaries} shows some preliminaries, including systems description, \texttt{STP} of matrices, as well as some elementary definitions. In Section
\ref{sec-structuralcontrollability}, the results on the structural controllability are presented with contribution I. The following three sections, i.e., Sections \ref{sec-network aggregation},\ref{sec-pin-controllability},\ref{sec-stability in probability}, investigate the network aggregation for the controllability of \texttt{BCNs}, the pinning controllability of \texttt{BNs}, and pinning stabilization in probability of \texttt{PBNs} with contributions II, III, IV, respectively. Finally, simulations are provided in Section \ref{sec-simulations}, while Section \ref{sec-conclusion} concludes this paper. Fig. \ref{fig-outline} provides the outline of this paper.
\begin{figure}[h!]
\centering
\includegraphics[width=0.48\textwidth=0.48]{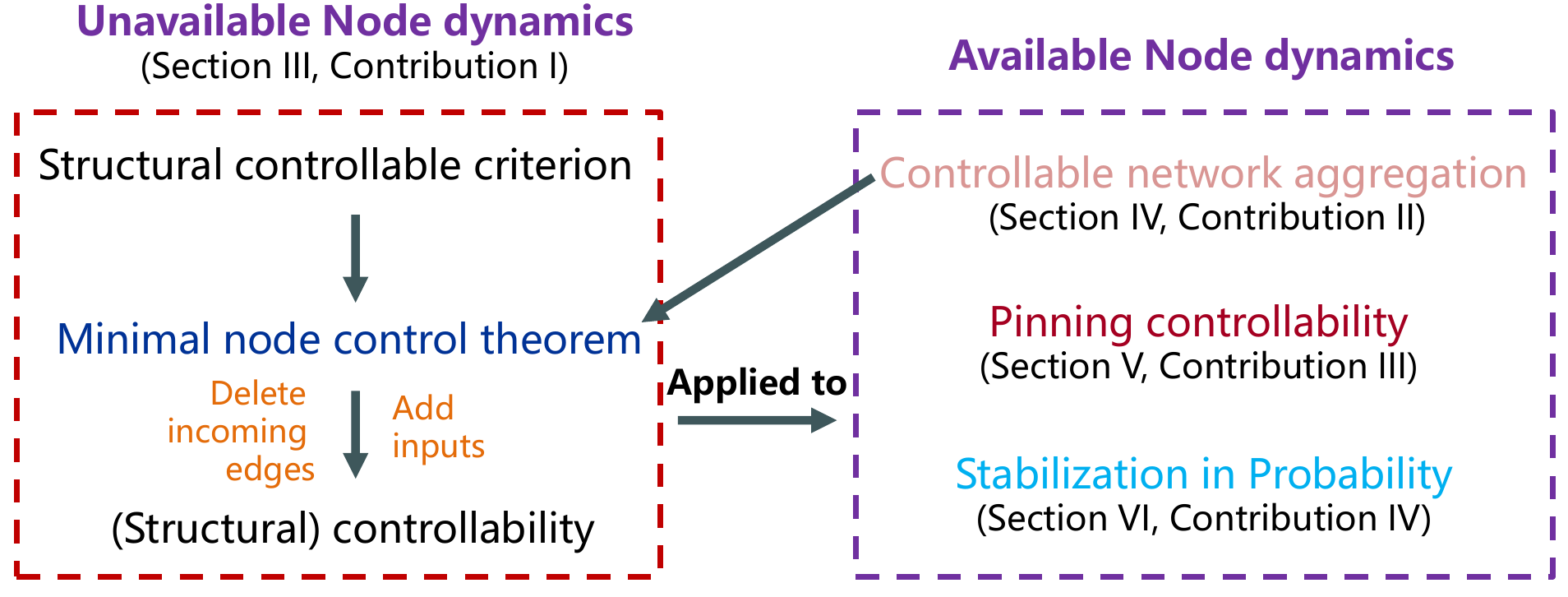}
\caption{The outline of this paper.}\label{fig-outline}
\end{figure}

\subsection{Notation and Terminology}
Throughout this paper, $\mathbb{R}^{m \times n}$ is the set of all $m \times n$-dimensional real matrices. Given integers $p$ and $q$ with $p<q$, the integer set $\{p,p+1,\cdots,q\}$ is denoted by $[p,q]$. Specially, let $\mathscr{D}:=\{1,0\}$. $I_{n}$ is used to represent the $n \times n$-dimensional identify matrix. By defining $\delta_n^i$ as the $i$-th canonical vector of size $n$, we can derive $\Delta_n=\{\delta_n^i\mid i\in[1,n]\}$. A matrix $B \in \mathbb{R}_{m \times n}$ is called a logical matrix if its $i$-th column, denoted by $\text{Col}_i(B)$, satisfies $\text{Col}_i(B) \in \Delta_{m}$ for any $i\in[1,n]$. Given set $S$, $\mid S \mid$ stands for its cardinal number. Given a matrix $B\in\mathbb{R}_{m\times n}$, if $n$ is divisible by $a$, we use $\text{Blk}_i(B)$ to represent the $i$-th $m\times a$-dimensional sub-matrix. For an event $E$, $\textbf{P}\{E\}$ means the probability of event $E$ occurring. The $n$-dimensional vector with all elements being one is represented by ${\bf 1}_n$. Given an $n$-dimensional column vector $v$ and a set $S\subseteq[1,n]$, the column vector $v_S$ is derived by deleting the rows with indices $i\in [1,n]\backslash S$. Particularly, $v_i:=v_{\{i\}}$.

\section{Preliminaries}\label{sec-preliminaries}
\subsection{System Description}
In this paper, we employ the related description for \texttt{BN} (\ref{e-LN}) and \texttt{BCN} (\ref{e-LCN}) which can also be referred to \cite{Azuma2015TCNS179,Azuma2019TCNS464}. Given \texttt{BN} (\ref{e-LN}), it can be expressed by $B({\bf G},{\bf F})$, where ${\bf G}$ and ${\bf F}$ respectively denote its network structure and component dynamics. More precisely,
\begin{enumerate}
  \item network structure ${\bf G}=({\bf V},{\bf E})$ is essentially a digraph, where ${\bf V}:=\{v_1,v_2,\cdots,v_n\}$ stands for the vertex set with the $i$-th node corresponding to vertex $v_i$. Ordered pair $(v_i,v_j)$ represents the oriented edge incoming $v_j$ from $v_i$, thus all the oriented edges are collected by ${\bf E}:=\{(v_i,v_j) \mid j\in[1,n],i\in{\bf X}_j\}$; and
  \item component dynamics ${\bf F}$ is used to represent a collection ${\bf F}:=({\bm f}_1,{\bm f}_2,\cdots,{\bm f}_n)$, where ${\bm f}_i$ is the dynamics of node $i$.
\end{enumerate}

As for \texttt{BCNs}, we only investigate the somewhat specific ones in the following form (\ref{equation-transformed-LCN}), and the general case (\ref{e-LCN}) can be viewed as a simple extension of it.
\begin{equation}\label{equation-transformed-LCN}
\left\{\begin{aligned}
{\bm x}_k(t+1)&=\tilde{{\bm f}}_k([{\bm x}_i(t)]_{i\in \tilde{{\bf X}}_k}),~k\in[m+1,n],\\
{\bm x}_j(t+1)&={\bm u}_j(t),~j \in [1,m].
\end{aligned}\right.
\end{equation}
Thus, unless otherwise stated, the \texttt{BCNs} that we refer to is of form (\ref{equation-transformed-LCN}). Considering \texttt{BCN} (\ref{equation-transformed-LCN}), it can be denoted by a triple $B(\tilde{{\bf G}},\tilde{{\bf F}},{\bf U})$, where ${\bf U}$ is the set of control inputs ${\bm u}_j$, $j\in[1,m]$. To make it easier to distinguish, digraphs ${\bf G}$ and $\tilde{{\bf G}}$ respectively represent the network structure of \texttt{BN} (\ref{e-LN}) and \texttt{BCN} (\ref{equation-transformed-LCN}). From then on, the triple $B(\tilde{{\bf G}},\tilde{{\bf F}},{\bf U})$ stands for \texttt{BCN} (\ref{equation-transformed-LCN}) except where otherwise stated.
\begin{remark}
It is worth noticing that network structure ${\bf G}$ is more concise than that in \cite{Azuma2015TCNS179,Azuma2019TCNS464}, where the activating in-neighbors and inhibiting in-neighbors are undistinguished by different edge-labeling functions in this paper. Thus, the analysis for \texttt{BNs} here requires less information of network structure than those in \cite{Azuma2015TCNS179,Azuma2019TCNS464}.
\end{remark}

\subsection{STP of Matrices}
Although matrix $L$ in (\ref{equ-assr-bn}) and matrix $\tilde{L}$ in (\ref{equ-assr-bcn}) have exponentially increasing dimensions, \texttt{STP} of matrices is indeed a powerful tool to transform the logical function into the corresponding multi-linear representation when considering the single logical function.
\begin{definition}[see \cite{chengdz2011springer}]\label{def-STP}
Given matrices $P\in\mathbb{R}_{a\times b}$ and $Q\in\mathbb{R}_{c\times d}$, the \texttt{STP} of matrices $P$ and $Q$ is defined as
\begin{equation*}
P \ltimes Q := (P\otimes I_{l/b}) (Q\otimes I_{l/c}),
\end{equation*}
where ``$\otimes$'' is the tensor (or, Kronecker) product of matrices and $l:=\text{lcm}\{b,c\}$ is the least common multiple of integers $b$ and $c$.
\end{definition}

Since the \texttt{STP} of matrices ``$\ltimes$'' is a generalization of the conventional matrix product, it provides a way to swap two matrices with any dimensions.
\begin{property}[see \cite{chengdz2011springer}]
For one thing, let vectors $U\in\Delta_{u}$ and $V\in\Delta_{v}$. Then there holds that $U \ltimes V = W_{[v,u]} \ltimes V \ltimes U$, where $W_{[v,u]}$ is called the swap matrix defined by $W_{[v,u]}=[I_v \otimes \delta_u^1, I_v \otimes \delta_u^2, \cdots, I_v \otimes \delta_u^u]$.

For the other thing, given a $p \times q$-dimensional matrix $B\in\mathscr{L}_{p \times q}$, there holds that $ U \ltimes B = (I_u \otimes B) \ltimes U$.
\end{property}

\begin{property}[see \cite{chengdz2011springer}]
If vector $U\in\Delta_{u}$, then $U \ltimes U = \Phi_u \ltimes U$, where $u$-valued power-reducing matrix $\Phi_u$ is defined as $\Phi_u=[\delta_u^1 \otimes \delta_u^1, \delta_u^2 \otimes \delta_u^2, \cdots, \delta_u^u \otimes \delta_u^u]$.
\end{property}

\begin{property}[see \cite{chengdz2011springer}]
If vectors $U,V\in\Delta_{2}$, then $V = \Psi \ltimes U \ltimes V$, where matrix $\Psi:={\bf 1}_2^\top \otimes I_2$ is called the left dummy matrix.
\end{property}

Define the canonical form of logical variable ${\bm x}_i\in\mathscr{D}$ as $x_i\mapsto\delta_2^{2-{\bm x}_i}$. By Lemma \ref{lemma-structurematrix}, the multi-linear form of arbitrary logical function ${\bm f}({\bm x}_1,{\bm x}_2,\cdots,{\bm x}_n): \mathscr{D}^n \rightarrow \mathscr{D}$ can be presented.
\begin{lemma}[See \cite{chengdz2011springer}]\label{lemma-structurematrix}
Given any function ${\bm f}({\bm x}_1,{\bm x}_2,\cdots,{\bm x}_n): \mathscr{D}^n\rightarrow\mathscr{D}$, there exists a unique logical matrix $L_{\bm f}\in\mathscr{L}_{2\times 2^n}$, called the structure matrix of function ${\bm f}$, such that ${\bm f}({\bm x}_1,{\bm x}_2,\cdots,{\bm x}_n)\mapsto L_{\bm f}\ltimes_{i=1}^{n}x_i$, where $\ltimes_{i=1}^{n}x_i:=x_1\ltimes x_2\ltimes \cdots\ltimes x_n$.
\end{lemma}

\subsection{Structural Concept of BNs}
In this subsection, we introduce structural controllability in the field of \texttt{BCNs} and some related notions.

To begin with, the controllability of \texttt{BCN} $B(\tilde{{\bf G}},\tilde{{\bf F}},{\bf U})$ is defined.

\begin{definition}[See \cite{akutsu2007jtb670}]\label{def-controllability}
\texttt{BCN} $B(\tilde{{\bf G}},\tilde{{\bf F}},{\bf U})$ is said to be controllable if, for any state pair $\alpha,\beta\in\mathscr{D}^n$, there exist an integer $N\geq 0$ and a control sequence
${\bf u}(1),{\bf u}(2),\cdots,{\bf u}(N-1)$ that steers this \texttt{BCN} from ${\bf x}(0)=\alpha$ to ${\bf x}(N)=\beta$.
\end{definition}

By following the definitions in \cite{margaliot2018auto56}, a state node $\vec{v}$ is said to be a channel if its out-degree is one and there does not exist a self loop $(\vec{v},\vec{v})$. The vertex representing the control input ${\bm u}_j$ is called a generator. Additionally, a control node is the one that directly connects with the generators; otherwise, it is called a simple node.

In order to define the structural controllability of \texttt{BCN} $B(\tilde{{\bf G}},\tilde{{\bf F}},{\bf U})$, the concepts of structural equivalence and closed \texttt{BNs} are introduced.
\begin{definition}[see \cite{Azuma2015TCNS179}]
Given \texttt{BN} $B({\bf G},{\bf F})$, \texttt{BN} $B({\bf G},\breve{{\bf F}})$ is said to be structurally equivalent to $B({\bf G},{\bf F})$, where the component dynamics $\breve{{\bf F}}$ may be different from ${\bf F}$.
\end{definition}

\begin{definition}[see \cite{Azuma2015TCNS179}]
Given \texttt{BCN} $B(\tilde{{\bf G}},\tilde{{\bf F}},{\bf U})$, \texttt{BN} $B(\hat{{\bf G}},\hat{{\bf F}})$ is called its closed \texttt{BN}, which is obtained by removing all the generators and their adjacency edges to controlled nodes.
\end{definition}

\begin{definition}[see \cite{Azuma2015TCNS179}]
Given \texttt{BCN} $B(\tilde{{\bf G}},\tilde{{\bf F}},{\bf U})$, \texttt{BCN} $B(\hat{{\bf G}},\hat{{\bf F}},{\bf U})$ is said to be structurally equivalent to \texttt{BCN} $B(\tilde{{\bf G}},\tilde{{\bf F}},{\bf U})$ if their closed \texttt{BNs} are structurally equivalent.
\end{definition}

On the basis of the above preparations, we provide the definition of the structural controllability for \texttt{BCN} $B(\tilde{{\bf G}},\tilde{{\bf F}},{\bf U})$.
\begin{definition}\label{def-structuralcontrollable}
Given $B(\tilde{{\bf G}},\tilde{{\bf F}},{\bf U})$, it is said to be structurally controllable if all its structurally equivalent \texttt{BCNs} are controllable.
\end{definition}

According to Definition \ref{def-structuralcontrollable}, the traditional methods in \cite{chengdz2009tac1659,zhaoy2010scl767,liangjl2017tac6012,zhuqx2018tac} are disabled, unless we try all kinds of logical couplings as in \cite{liht2020scis}. However, the time complexity will have reached $\Theta((2^{n+m})^2)$, which is too heavy to address a practical large-scale \texttt{BCN}.

\section{Structural Controllability of BNs}\label{sec-structuralcontrollability}
In this section, we investigate the necessary and sufficient criterion for the structural controllability of \texttt{BCN} $B(\tilde{{\bf G}},\tilde{{\bf F}},{\bf U})$, and then analyze the time complexity of Problem \ref{problem-minimalnodecontrol} for an arbitrary \texttt{BN} $B({\bf G},{\bf F})$.
\begin{definition}[See \cite{margaliot2018auto56}]\label{def-ctr-path}
Given \texttt{BCN} $B(\tilde{{\bf G}},\tilde{{\bf F}},{\bf U})$, a non-empty ordered set $\Omega\subseteq {\bf G}$ is said to be a controlled path, if the following two conditions are both satisfied:
\begin{itemize}
  \item[1)] the first entry of $\Omega$ is a generator; and
  \item[2)] for $\mid \Omega \mid>1$, the $i$-th element of $\Omega$, $i\in[2,\mid \Omega \mid]$, is a state node and it is the only out-neighbor of the ($i-1$)-th element of $\Omega$.
\end{itemize}
Two controlled paths $\Omega_i$ and $\Omega_j$ are said to be disjoint if they do not have any overlapping node.
\end{definition}

\begin{definition}[See \cite{margaliot2018auto56}]\label{def-layernetwork}
Given an acyclic diagraph $\tilde{{\bf G}}=(\tilde{{\bf V}},\tilde{{\bf E}})$, it is called a $k$-layer graph if, each $v\in\tilde{{\bf V}}$ lies in a single layer $L_h$, $h\in[1,k]$, and for such $v\in L_h$, $(v,v')\in \tilde{{\bf E}}$ implies $v' \in L_{h+1}$.
\end{definition}

In \cite{margaliot2018auto56}, Weiss {\em et al.} obtained a necessary and sufficient condition for the controllability of \texttt{CBNs} as that a \texttt{CBN} is controllable if and only if its network structure can be split into several disjoint controlled paths. Nevertheless, this condition is indeed not sufficient for the structural controllability of \texttt{BCNs}. Hence, we need a more strict condition. Likewise, although the minimum node control problem has been analyzed to be NP-hard for \texttt{CBNs} \cite{margaliot2018auto56}, the complexity of Problem \ref{problem-minimalnodecontrol} for \texttt{BN} $B({\bf G},{\bf F})$ needs to be further analyzed, due to the enhancement of condition.

\subsection{A Polynomial-Time Condition for Structural Controllability}
In virtue of the controlled paths in Definition \ref{def-ctr-path}, a necessary and sufficient condition is established for the structural controllability of \texttt{BCN} $B(\tilde{{\bf G}},\tilde{{\bf F}},{\bf U})$.
\begin{theorem}\label{theorem-structuralcontrollability}
Consider \texttt{BCN} $B(\tilde{{\bf G}},\tilde{{\bf F}},{\bf U})$. The following three statements are mutually equivalent:
\begin{itemize}
  \item[1)] \texttt{BCN} $B(\tilde{{\bf G}},\tilde{{\bf F}},{\bf U})$ is structurally controllable;
  \item[2)] digraph $\tilde{{\bf G}}$ can be decomposed into several disjoint root in-trees\footnote{An acyclic digraph $\tilde{{\bf G}}$ is called a root in-tree if, there is a unique vertex $v$ with out-degree zero, and the out-degree of all other vertices is one.} whose leaves are all generators;
  \item[3)] digraph $\tilde{{\bf G}}$ is acyclic, and \textbf{(${\bf C1}$) the in-neighbor set of every simple node in $\tilde{{\bf G}}$ is non-empty and only contains channels}.
\end{itemize}
\end{theorem}
\begin{proof}
We would prove this theorem by implementing ${\bf 1)}\Rightarrow {\bf 3)}\Rightarrow {\bf 2)}\Rightarrow {\bf 1)}$.

First of all, we verify the implement ${\bf 1)}\Rightarrow {\bf 3)}$. Assuming that digraph $\tilde{{\bf G}}$ contains a directed cycle $C$ with length $p$, denoted by $C=\{v_{m+1},v_{m+2},\cdots,v_{m+p}\}\subseteq \tilde{{\bf V}}$ without loss of generality, one has that $v_{m+i}$, $i\in[1,p]$, are neither a generator nor a controlled node, because a generator does not have any in-neighbor and a controlled node has a unique in-neighbor as a generator.

Taking \texttt{BCN} $B(\tilde{{\bf G}},\vec{{\bf F}},{\bf U})$ into account, where the logical couplings in $\vec{{\bf F}}:=(\vec{{\bm f}_1},\vec{{\bm f}_2},\cdots,\vec{{\bm f}_n})$ are uniquely ``$\vee$''. Let initial state ${\bf x}(0)=\Sigma_{i=1}^{p}\delta_{n}^{m+i}$, it is obvious that \texttt{BCN} $B(\tilde{{\bf G}},\vec{{\bf F}},{\bf U})$ cannot be steered from the above ${\bf x}(0)$ to the target state $[0,0,\cdots,0]^\top$. Thus, it contradicts with the assumption that \texttt{BCN} $B(\tilde{{\bf G}},\tilde{{\bf F}},{\bf U})$ is structurally controllable.
\begin{figure}[h!]
\centering
\includegraphics[width=0.45\textwidth=0.48]{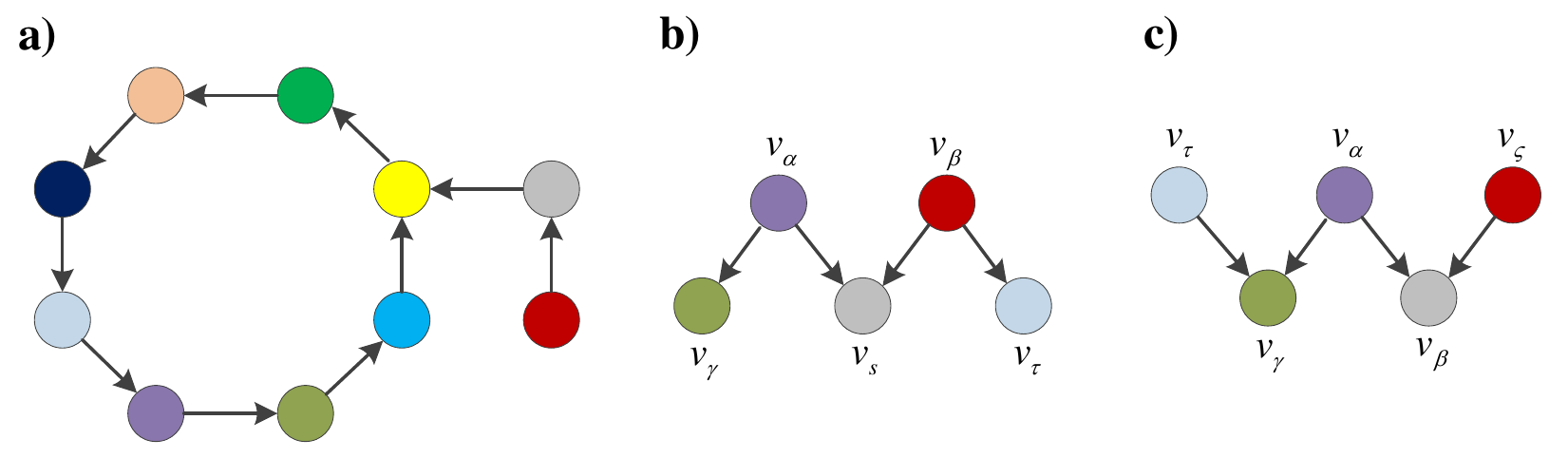}
\caption{Subgraph a) represents the case of cyclic $\tilde{{\bf G}}$, while subgraph b) stands for the scene that $\tilde{{\bf G}}$ contains a vertex $v_s$, whose in-neighbors are neither channels nor generators. Moreover, subgraph c) means the situation when there is a vertex $v_\beta$ that has an in-neighbor not as a channel or a generator.}\label{fig-cycle}
\end{figure}
If there is a node $v$ with zero in-degree, then the state of node $v$ must be constant to $0$ or $1$.

Next, we assume that \texttt{BCN} $B(\tilde{{\bf G}},\tilde{{\bf F}},{\bf U})$ is structurally controllable but graph $\tilde{{\bf G}}$ contains a simple node $v_s$, whose in-neighbor set does not contain a channel as shown in Fig. \ref{fig-cycle}b). We derive the functions $\bar{{\bf F}}:=\{\bar{{\bm f}}_1,\bar{{\bm f}}_2,\cdots,\bar{{\bm f}}_n\}$ from replacing all the functional variables of $\tilde{{\bm f}}_i$ by ``$\wedge$'' and consider \texttt{BCN} $B(\tilde{{\bf G}},\bar{{\bf F}},{\bf U})$.

Since \texttt{BCN} $B(\tilde{{\bf G}},\tilde{{\bf F}},{\bf U})$ is structurally controllable, \texttt{BCN} $B(\tilde{{\bf G}},\bar{{\bf F}},{\bf U})$ is controllable; and it is reachable from $\mu=[0,0,\cdots,0]^\top$ to $\nu={\bf 1}_{n}^\top-(\delta_n^s)^\top$. Without loss of generality, we assume that vertex $v_s$ only has two in-neighbors $v_\alpha$ and $v_\beta$. Because $v_\alpha$ and $v_\beta$ are not channels, one can find vertices $v_\gamma \neq v_s$ and $v_\tau \neq v_s$ satisfying $(v_\alpha,v_\gamma)\in \tilde{{\bf E}}$ and $(v_\beta,v_\tau)\in \tilde{{\bf E}}$. Thus, digraph $\tilde{{\bf G}}$ contains a self loop $(v_s,v_s)$; it is a contradiction with the former conclusion that $\tilde{{\bf G}}$ does not contain an oriented cycle.

Furthermore, we prove that the in-neighbor set of each simple node can only contain channels. Seek a contradiction. Without loss of generality, we suppose that digraph $\tilde{{\bf G}}$ contains a node
$v_\alpha$ that has two out-neighbors $v_\beta$ and $v_\gamma$, i.e., $(v_\alpha,v_\beta)\in\tilde{{\bf E}}$ and $(v_\alpha,v_\gamma)\in\tilde{{\bf E}}$. As proved above, there must be another two channels $v_\varsigma$ and $v_\tau$ that act as the in-neighbor of vertices $v_\beta$ and $v_\gamma$, respectively; it means that $(v_\varsigma,v_\beta)\in\tilde{{\bf E}}$ and $(v_\tau,v_\gamma)\in\tilde{{\bf E}}$. Let the logical operators in functions $\tilde{{\bm f}}_\beta$ and $\tilde{{\bm f}}_\gamma$ be full with ``$\wedge$'' and ``$\vee$'' respectively, and consider the system state $\mu\in \mathscr{D}^n$ with $\mu_\beta=1$ and $\mu_\gamma=0$. If state $\mu$ is reachable from any other state $\nu\in\mathscr{D}^n$ by one step, then $\nu_\alpha=1$ due to the conjunctive function ${\bm f}_\beta$. However, if $\nu_\alpha=1$, then it deduces that $\mu_\gamma=1$, because of the disjunctive function $\tilde{{\bm f}}_\gamma$. This is a contradiction. Therefore, if \texttt{BCN} $B(\tilde{{\bf G}},\tilde{{\bf F}},{\bf U})$ is structurally controllable, then the in-neighbor set of each simple node can only contain channels.

Afterwards, we verify the implement that ${\bf 3)}\Rightarrow {\bf 2)}$. Divide digraph $\tilde{{\bf G}}$ into several maximum weakly connected subgraphs $\tilde{{\bf G}}^i$, $i\in[1,\omega]$. Because digraph $\tilde{{\bf G}}^i$ is acyclic, we can find the topological sort, denoted by $Q(\tilde{{\bf G}}^i)$, for its nodes. For the terminal vertex $v_d$ in the sequence $Q(\tilde{{\bf G}}^i)$, its out-degree must be equal to zero. Subsequently, we prove that vertex $v_d$ is the unique node with out-degree zero in $\tilde{{\bf G}}^i$. Suppose that there is another vertex $v_{\tilde{d}}$ with out-degree zero. Since the digraph $\tilde{{\bf G}}^i$ is weakly connected, there must exist a vertex $v_\tau$ which can reach both $v_d$ and $v_{\tilde{d}}$. Along with the oriented path from $v_\tau$ to $v_d$, we can find a vertex with out-degree more than two. It contradicts with the fact that the in-neighbors of each simple node are channels. Meanwhile, it also concludes that vertex $v_d$ is reachable from other vertices in $\tilde{{\bf G}}^i$, since vertex $v_d$ is the unique vertex in acyclic digraph $\tilde{{\bf G}}^i$ with out-degree zero. Thus, each subgraph $\tilde{{\bf G}}^i$ is in a shape of a rooted in-tree. As the in-neighbor set of each simple node is nonempty, the leaves of each in-tree must be generators.

Finally, we show the deduction ${\bf 2)\Rightarrow 1)}$. Suppose that the digraph $\tilde{{\bf G}}$ is composed of a series of vertices-disjoint root trees ${\bf R}^1$, ${\bf R}^2$, $\cdots$, ${\bf R}^\lambda$ with leaves being generators. If we can prove the structural controllability for the case of $\lambda=1$, then the situation of $\lambda\geq2$ can be regarded as a straightforward consequence.

Consider the case that $\tilde{{\bf G}}$ is a root in-tree. With loss of generality, we assume that the lengths of all the paths from generators to its root are the same as $l$. Otherwise, we can unify the length by adding some virtual simple nodes. Then, by feeding the control inputs one by one like a shift register, we can achieve the reachability from arbitrary $\mu\in\mathscr{D}^n$ to arbitrary $\nu\in\mathscr{D}^n$. Thus, \texttt{BCN} $B(\tilde{{\bf G}},\tilde{{\bf F}},{\bf U})$ is controllable. Considering an arbitrary node $v_j$ in layer $L_k$ with $\nu_j\in\mathscr{D}$, we can find a group of variables $\beta_i$, $i\in \tilde{{\bf X}}_j$ satisfying $\nu_j=\tilde{{\bm f}}_j([\beta_i]_{i\in \tilde{{\bf X}}_j})$. Such procedure can be executed until reaching the generators. The tree shape guarantees that the state nodes in the same layer can be considered dependently. The proof is established.
\end{proof}

\begin{remark}
As for the controllability of \texttt{CBNs}, it only requires that network structure is acyclic and each vertex has a channel or a generator as one of in-neighbors. Obviously, condition ({\bf C1}) is stronger than that for \texttt{CBNs}, thus it of course can also address the \texttt{CBNs} as a sufficient condition. Besides, the time complexity of this criterion is remarkably lower than $\Theta(2^{3n+3m})$ in \cite{liht2020scis}, which is computationally heavy for large-scale \texttt{BCNs}.
\end{remark}

\begin{remark}
Compared with the criterion in \cite{margaliot2018auto56}, the conditions 2) and 3) are obviously stronger. It leads to that the proof of NP-hardness for Problem \ref{problem-minimalnodecontrol} to an arbitrarily given \texttt{CBN} in \cite{margaliot2018auto56} is not suitable anymore.
\end{remark}

\begin{remark}
As for the traditional case when the node dynamics can be identifiable or known, Theorem \ref{theorem-structuralcontrollability} also plays an efficient role to check the controllability as a sufficient condition; it is the cornerstone for the works in Section \ref{sec-network aggregation}, Section \ref{sec-pin-controllability} and Section \ref{sec-stability in probability}.
\end{remark}

In \cite{margaliot2012aut1218}, Laschov and Margaliot proposed a stronger notion, termed as fixed-time controllability, than controllability, which is important in biological engineering. For instance, for a multi-cellular organism composed of several identical cell-cycles where subnetworks are modeled as the same \texttt{BCNs}, sometimes we may be interested in synchronizing all partitions via control inputs. It should be noticed that if the fixed-time controllability can be achieved, so is the fixed-time synchronization. Moreover, the trajectory controllability is to define whether \texttt{BCNs} can track an arbitrarily given state sequence. In the following, we would also show that the states of root nodes can be trajectory controllable.
\begin{definition}\label{defn-fixed-time-controllability}
\texttt{BCN} $B(\tilde{{\bf G}},\tilde{{\bf F}},{\bf U})$ is said to be structurally fixed-time controllable if, for every \texttt{BCN} that is structurally equivalent to \texttt{BCN}
$B(\tilde{{\bf G}},\tilde{{\bf F}},{\bf U})$, there exist an integer $\tau>0$ and an input sequence ${\bf u}(0)$, ${\bf u}(1)$, $\cdots$, ${\bf u}(\tau-1)$ to steer it from ${\bf x}(0)=\mu$ to ${\bf x}(\tau)=\nu$ for any $\mu,\nu\in \mathscr{D}^n$.
\end{definition}

\begin{definition}\label{defn-n-fixed-time-controllability}
\texttt{BCN} $B(\tilde{{\bf G}},\tilde{{\bf F}},{\bf U})$ is said to be structurally $\eta$ fixed-time controllable if, there is an integer $\eta>0$ such that for every \texttt{BCN} that is structurally equivalent to \texttt{BCN} $\Sigma(\tilde{{\bf G}},\tilde{{\bf F}},{\bf U})$, one can find an input sequence ${\bf u}(0)$, ${\bf u}(1)$, $\cdots$, ${\bf u}(\eta-1)$ to steer it from ${\bf x}(0)=\mu$ to ${\bf x}(\eta)=\nu$ for any $\mu,\nu\in \mathscr{D}^n$.
\end{definition}

\begin{corollary}\label{cor-fixed-time}
\texttt{BCN} $B(\tilde{{\bf G}},\tilde{{\bf F}},{\bf U})$ is structurally fixed-time controllable (or structurally $\eta$ fixed-time controllable) if and only if it is structurally controllable. The minimal number $\eta$ is the layer number of these root in-trees.
\end{corollary}
\begin{proof}
As already reported in \cite{margaliot2012aut1218}, for \texttt{BCN} $B(\tilde{{\bf G}},\tilde{{\bf F}},{\bf U})$ that is structurally $\eta$ fixed-time controllable, it is also $T$ fixed-time controllable with $T>\eta$. Thereby the structural fixed-time controllability and structural $\eta$ fixed-time controllability are equivalent. We can observe that the design of input variables in the proof of deduction ${\bf 2}) \Rightarrow {\bf 1})$ in Theorem \ref{theorem-structuralcontrollability} is independent of the initial state. Thus, one can always realize the reachability within $l$ time steps, which is equal to the layer number of root in-trees.
\end{proof}

\begin{definition}\label{defn-trajectory-controllability}
\texttt{BCN} $B(\tilde{{\bf G}},\tilde{{\bf F}},{\bf U})$ is said to be trajectory controllable w.r.t. state nodes $\{v_j\mid j\in\ddot{\bf V}\} \subseteq \tilde{{\bf V}}$ if, for any given binary sequence $S_{\infty}=\{s_1,s_2,\cdots\}$ with $s_i\in \mathscr{D}^{\mid \ddot{\bf V} \mid}$, there exists an input sequence ${\bf u}(0), {\bf u}(1), \cdots$ such that ${\bf x}_{\ddot{\bf V}}(N)=s_1, {\bf x}_{\ddot{\bf V}}(N+1)=s_2, \cdots$ for certain integer $N$.
\end{definition}

\begin{definition}\label{defn-trajectory-structuralcontrollability}
\texttt{BCN} $B(\tilde{{\bf G}},\tilde{{\bf F}},{\bf U})$ is said to be structurally trajectory controllable w.r.t. state nodes $\{v_j\mid j\in\ddot{\bf V}\} \subseteq \tilde{{\bf V}}$ if, for any given binary sequence $S_{\infty}=\{s_1,s_2,\cdots\}$ with $s_i\in \mathscr{D}^{\mid \ddot{\bf V} \mid}$, all structurally equivalent \texttt{BCNs} are trajectory controllable.
\end{definition}

\begin{corollary}\label{cor-trajectory-controllability}
\texttt{BCN} $\Sigma(\tilde{{\bf G}},\tilde{{\bf F}},{\bf U})$ is structurally trajectory controllable w.r.t. its root nodes if it is structurally controllable.
\end{corollary}
\begin{proof}
This conclusion can be implied by the design of input variables in the proof of deduction ${\bf 2}) \Rightarrow {\bf 1})$ in Theorem \ref{theorem-structuralcontrollability}.
\end{proof}

\begin{remark}
Theorem \ref{theorem-structuralcontrollability} and Corollary \ref{cor-fixed-time} reveal a phenomenon that \texttt{BCN} $\Sigma(\tilde{{\bf G}},\tilde{{\bf F}},{\bf U})$ is structurally controllable if and only if it is
structurally fixed-time controllable. This is pretty different with the conclusion in \cite{margaliot2012aut1218}.
\end{remark}

\begin{remark}
Corollary \ref{cor-fixed-time} and Corollary \ref{cor-trajectory-controllability} are visual from our results. Compared with the traditional \texttt{STP} approach in \cite{liht2020scis}, Theorem \ref{theorem-structuralcontrollability} provides a more visual way to analyze the structural fixed-time controllability of \texttt{BCNs}, except for the time complexity.
\end{remark}

\begin{remark}
Furthermore, it is worthwhile noticing that the toolbox provided by Cheng and his colleagues based on MATLAB can only deal with the controllability analysis of \texttt{BCNs} based on \texttt{STP} of matrices with $n\leq 25$
or so, as pointed out in \cite{zhaoy2010scl767}. As a contract, the structural criterion-Theorem \ref{theorem-structuralcontrollability}-can be checked within $\Theta(n^2)$ time, as provided in Algorithm \ref{algorithm-controllability}. It means that our approach is computationally efficient even for large-scale \texttt{BCNs}.
\end{remark}

\begin{algorithm}[h!]
\caption{Testing the structural controllability of \texttt{BCN} $B(\tilde{{\bf G}},\tilde{{\bf F}},{\bf U})$ with $n$ state variables and $m$ control inputs.}\label{algorithm-controllability}
\begin{algorithmic}[1]
\Require The graph $\tilde{{\bf G}}=(\tilde{{\bf V}},\tilde{{\bf E}})$.
\Ensure Whether \texttt{BCN} $B(\tilde{{\bf G}},\tilde{{\bf F}},{\bf U})$ is structurally controllable?
\If{$\tilde{{\bf G}}$ contains a cycle}
\State\Return{$B(\tilde{{\bf G}},\tilde{{\bf F}},{\bf U})$ is not structurally controllable}
\Else
\State create an $n$-bits list $L$ with initial values $0$
\If{$\mid \mathbf{N}^{\text{out}}_i\mid\equiv 1$}
\For{$v_j\in \mathbf{N}^{\text{out}}_i$}
\State $L(j)\leftarrow L(j)+1$
\EndFor
\EndIf
\EndIf
\If{$L(j)=\mid\mathbf{N}^{\text{in}}_j\mid$ for $j\in[1,n]$}
\State\Return{\texttt{BCN} $B(\tilde{{\bf G}},\tilde{{\bf F}},{\bf U})$ is structurally controllable}
\EndIf
\end{algorithmic}
\end{algorithm}	

\subsection{Complexity Analysis of Minimum Node Control Problem}
In order to prove the NP-hardness of Problem \ref{problem-minimalnodecontrol}, we prove a somewhat stronger result: the minimum node control problem for the structural controllability of $3$-layer \texttt{BCNs} with customized $\breve{{\bf G}}$ is NP-hard. These $3$-layer \texttt{BCNs} were firstly proposed in \cite{margaliot2018auto56} and network structure $\breve{{\bf G}}$ satisfies that every vertex in the layer $1$ is a generator jointing to a vertex in layer $2$ and every vertex in layer $2$ has the unique in-neighbor lying in the layer $1$. With regard to such \texttt{BCN}, we can consider its structural controllability by directly applying Theorem \ref{theorem-structuralcontrollability}.
\begin{corollary}\label{cor-3layer}
Consider a $3$-layer \texttt{BCN} with the above network structure $\check{{\bf G}}$. It is structurally controllable if and only if the out-degree of each vertex in layer $2$ is no more than $1$.
\end{corollary}

To proceed, once polynomial-time Algorithm \ref{algorithm-controllability} has been provided, what we need to do is to give a polynomial-time reduction from an NP-hard problem to Problem \ref{problem-minimalnodecontrol}  for the above $3$-layer \texttt{BCNs}.
\begin{definition}[see \cite{minimumvertexcover}]\label{def-vertexcover}
Given an undirected graph $G=(V,E)$, vertex subset $\tilde{V} \subseteq V$ is called a vertex cover of graph $G$ if, for every undirected edge $e_{ij}\in E$, exact one of $i \in \tilde{V}$ or $j \in \tilde{V}$ holds.
\end{definition}

\begin{problem}[see \cite{minimumvertexcover}]\label{problem-min-vertex-cover}
Given an undirected graph $G=(V,E)$, the minimum vertex cover problem is to find the minimum vertex subset $\tilde{V}^\ast$ for $G$.
\end{problem}

\begin{lemma}[see \cite{minimumvertexcover}]
Problem \ref{problem-min-vertex-cover} is NP-hard.
\end{lemma}

\begin{theorem}\label{thm-nphard}
Given \texttt{BN} $B({\bf G},{\bf F})$, Problem \ref{problem-minimalnodecontrol} is NP-hard.
\end{theorem}
\begin{proof}
Algorithm \ref{algorithm-controllability} has been developed to check whether a \texttt{BCN} is structurally controllable in a polynomial time. Thus, this problem is in P (the structural controllability of a given \texttt{BCN} can be checked in a polynomial time); and it suffices to complete the proof by exhibiting a polynomial-time reduction from NP-hard Problem \ref{problem-min-vertex-cover}.

Our proof proceeds to prove a somewhat stronger result that the minimum node control problem for a $3$-layer \texttt{BCN}, given in Corollary \ref{cor-3layer}, is NP-hard.

In Problem \ref{problem-min-vertex-cover}, we are given an undirected graph $G=(V,E)$. Then a digraph $\hat{{\bf G}}=(\hat{{\bf V}},\hat{{\bf E}})$ is constructed to guarantee that once Problem \ref{problem-minimalnodecontrol} of $3$-layer \texttt{BCNs} with networks structure $\hat{{\bf G}}$ is solved, so is Problem \ref{problem-min-vertex-cover} for undirected graph $G$. We now provide the polynomial time reduction. For every self loop $(v_i,v_i)$, we add a new node $v_i'$ and collect them by the set $V'$. Define $\hat{{\bf V}}_2=V\cup V'$ and $\hat{{\bf V}}_3=E$, then let $\hat{{\bf V}}=\hat{{\bf V}}_2 \cup \hat{{\bf V}}_3$. Afterwards, we construct the directed edge set $\hat{{\bf E}}$. For every vertex $e_{ij}\in \hat{{\bf V}}_3$, one induces a self loop $(e_{ij},e_{ij})\in \hat{{\bf E}}_1$. Besides, for each $v_i,v_j\in\hat{{\bf V}}_2$, we assign the directed edges $(e_{ij},v_i)\in \hat{{\bf E}}_2$ and $(e_{ij},v_j)\in \hat{{\bf E}}_2$. Particularly, if self loop $(v_i,v_i)\in E$, then it is enough to assign directed edges $(e_{ii},v_i)\in \hat{{\bf E}}_2$ and $(e_{ii},v'_i)\in \hat{{\bf E}}_2$. Let $\hat{{\bf E}}=\hat{{\bf E}}_1 \cup \hat{{\bf E}}_2$.
\begin{figure}[h!]
\centering
\includegraphics[width=0.4\textwidth=0.45]{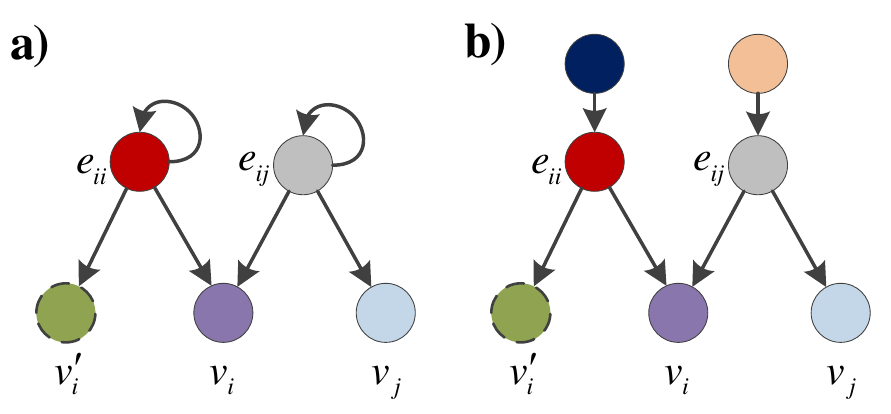}
\caption{Subgraph a) gives the construction of digraph $\hat{{\bf G}}$ from $G$, and subgraph b) presents the network structure of $3$-layer \texttt{BCN} after controlling nodes in $\hat{{\bf V}}_3$.}\label{fig-reductions}
\end{figure}

In the case that the \texttt{BN} with the above network structure $\hat{{\bf G}}$. We further consider Problem \ref{problem-minimalnodecontrol}. Because vertex $e_{ij}\in \hat{{\bf V}}_3$ has a self loop, it must be selected as a controlled node. Suppose that set $\Lambda^\ast$ is the minimum controlled node set of the \texttt{BN} with network structure $\hat{{\bf G}}$, then we can find another solution of $\vec{\Lambda}^\ast$ satisfying $\vec{\Lambda}^\ast \cap V' =\emptyset$ and $\mid \vec{\Lambda}^\ast \mid=\mid\Lambda^\ast\mid$ via replacing the $v'_i$ in $\Lambda^\ast$ by $v_i$. We then prove that $\tilde{\Lambda}=\vec{\Lambda}^\ast \backslash \hat{{\bf V}}_3$ is exactly the minimum vertex cover of undirected graph $G=(V,E)$.

For one thing, we illustrate that $\tilde{\Lambda}$ is a vertex cover of undirected graph $G$. Notice that every vertex in layer $2$ has out-degree $2$, if we can control the vertices in $\vec{\Lambda}^\ast\backslash \hat{{\bf V}}_3$ so as to control the out-degree of vertices in layer $2$, i.e., $\hat{{\bf V}}_3$, to be either $1$ or $0$. Therefore, for each vertex $e_{ij}\in \hat{{\bf V}}_2$, it implies $v_i\in \vec{\Lambda}^\ast\backslash \hat{{\bf V}}_3$ or $v_j\in \vec{\Lambda}^\ast\backslash \hat{{\bf V}}_3$. It also holds for self loop $(v_i,v_i)\in E$. Therefore, set $\tilde{\Lambda}$ satisfies Definition \ref{def-vertexcover} and $\tilde{\Lambda}$ is a vertex cover of graph $G$.

For the other thing, we prove that the set $\tilde{\Lambda}$ is minimal. Seek a contradiction. Suppose that there is another set $\hat{\Lambda}$ with $\mid \hat{\Lambda} \mid < \mid\tilde{\Lambda} \mid$ such that it is also a vertex cover of graph $G$. Thus, if we control the nodes in $\hat{\Lambda} \cup \hat{{\bf V}}_3$, then it can be checked that the \texttt{BCN} after control is structurally controllable. Moreover, one derives $\hat{\Lambda} \cap \hat{{\bf V}}_3 = \emptyset$ and $\tilde{\Lambda} \cap \hat{{\bf V}}_3=\emptyset$, which implies $\mid\hat{\Lambda} \cup \hat{{\bf V}}_3\mid < \mid\tilde{\Lambda} \cup \hat{{\bf V}}_3\mid=\mid\vec{\Lambda}^\ast\mid$. It is a contradiction because set $\vec{\Lambda}^\ast$ is the minimal node set.
\end{proof}

The results in Theorem \ref{theorem-structuralcontrollability} would be used to overcome some other difficult control problems, including network aggregation w.r.t. controllability of \texttt{BNs}, pinning controllability of \texttt{BNs} as well as stabilization in probability of \texttt{PBNs}. These results will be respectively displayed in Section \ref{sec-network aggregation}, Section \ref{sec-pin-controllability} and Section \ref{sec-stability in probability}.

\section{Network Aggregation With Regard to Controllability}\label{sec-network aggregation}
As mentioned in \cite{zhaoy2015tnnls}, Zhao {\em et al.} formulated that ``any meaningful and efficient sufficient condition for controllability has not yet be found''. Moreover, they also pointed out that this is inherently a
difficult problem and an open problem. The main reason is that some state variables may act as the input variables of another subnetwork, but their values cannot be arbitrarily injected as generators. Compared with the existing results on controllability, the root nodes of a structurally controllable \texttt{BCN} are also trajectory controllable; it guarantees the potential possibility to provide a network aggregation w.r.t. controllability.

In this section, we will investigate a network aggregation approach, which not only addresses the problems ({\bf P1}) and ({\bf P5}) to some extent but also gives an approach to solve Problem \ref{problem-minimalnodecontrol}.

\subsection{Network Aggregation for Controllability}
Given \texttt{BCN} $B(\tilde{{\bf G}},\tilde{{\bf F}},{\bf U})$ with $\tilde{{\bf G}}=(\tilde{{\bf V}},\tilde{{\bf E}})$, set $\mathscr{X}=\{v_1,v_2,\cdots,v_n\}$ and set $\mathscr{U}=\{u_1,u_2,\cdots,u_m\}$ are respectively used to denote the sets of state nodes and generators.
\begin{definition}[see \cite{zhangkz2020tac}]
Given \texttt{BCN} $B(\tilde{{\bf G}},\tilde{{\bf F}},{\bf U})$, the partition $\tilde{{\bf V}}=\mathbf{N}_1\cup\mathbf{N}_2\cup\cdots\cup\mathbf{N}_\omega$ is said to be a network aggregation if, $\mathbf{N}_i$ is a nonempty subset of $\mathscr{X}\cup\mathscr{U}$ and $\mathbf{N}_i \cap \mathbf{N}_j=\emptyset$, for any $i,j\in[1,\omega]$.
\end{definition}

By viewing set $\mathbf{N}_i$, $i\in[1,\omega]$, as a supper node, we can construct the aggregated graph $\tilde{{\bf G}}_A=(\tilde{{\bf V}}_A,\tilde{{\bf E}}_A)$, where $\tilde{{\bf V}}_A=\{\mathbf{N}_i\mid i\in[1,\omega]\}$, and $(\mathbf{N}_i,\mathbf{N}_j)\in \tilde{{\bf E}}_A$ if and only if one has $(v_{k_i},v_{k_j})\in \tilde{{\bf E}}$ for some $v_{k_i}\in \mathbf{N}_i$ and $v_{k_j}\in\mathbf{N}_j$. With regard to each supper node $\mathbf{N}_i$, the corresponding \texttt{BCN} is briefly termed as subnetwork $B(\tilde{{\bf G}}_i,\tilde{{\bf F}}_i,{\bf U}_i)$. More precisely, the state node set of subnetwork $B(\tilde{{\bf G}}_i,\tilde{{\bf F}}_i,{\bf U}_i)$ is given as $\mathbf{N}_i \cap \mathscr{X}$, while the input nodes are collected by $\mathscr{U}_i \cup \mathscr{U}'_i$, where $\mathscr{U}_i=\mathbf{N}_i \cap \mathscr{U}$ and $\mathscr{U}'_i=\{v\in\mathscr{X}\backslash\mathbf{N}_i \mid \exists v'\in \mathbf{N}_i\cap \mathscr{X},~\text{s.t.}~(v,v')\in \tilde{{\bf E}}\}$.

\begin{remark}
The most common network aggregation is the acyclic aggregation based on strongly connected components; it means that digraph $\tilde{{\bf G}}_A$ is acyclic \cite{zhangkz2020tac,zhaoy2015tnnls,zhaoyin2013tac1976}. However, the general acyclic network aggregation is obviously not suitable for controllability.
\end{remark}

In addition, according to Theorem \ref{theorem-structuralcontrollability}, the acyclic condition in digraph $\tilde{{\bf G}}_A$ seems to be prerequisite. Thus, we strengthen the traditional acyclic network aggregation so as to make it appropriate for controllability analysis.
\begin{definition}\label{def-aggregation}
Given \texttt{BCN} $B(\tilde{{\bf G}},\tilde{{\bf F}},{\bf U})$, the above acyclic aggregation $\mathbf{N}=\mathbf{N}_1\cup\mathbf{N}_2\cup\cdots\cup\mathbf{N}_\omega$ is called an channel aggregation if, there is an oriented edge $(v_\mu,v_o)\in\tilde{{\bf E}}$ with $v_\mu\in\mathbf{N}_i$ and $v_o\in\mathbf{N}_j$, such that $(v_{\tilde{\mu}},v_{\tilde{o}})\in \tilde{{\bf E}}$ and $v_{\tilde{\mu}}\in{\bf N}_i$ can imply $v_{\tilde{o}}\in\mathbf{N}_j$. On this basis, if $v_{\mu}=v_{\tilde{\mu}}$ can obtain $v_{\tilde{o}}=v_{o}$, then it is said to be a single-source-channel aggregation.
\end{definition}

\begin{theorem}
Given \texttt{BCN} $B(\tilde{{\bf G}},\tilde{{\bf F}},{\bf U})$ and a single-source-channel network aggregation $\mathbf{N}=\mathbf{N}_1\cup\mathbf{N}_2\cup\cdots\cup\mathbf{N}_\omega$, the entire \texttt{BCN} is controllable if, its root subnetworks are controllable and other subnetworks are structurally controllable.
\end{theorem}
\begin{proof}
Without loss of generality,  we assume that subnetwork $B(\tilde{{\bf G}}_\omega,\tilde{{\bf F}}_\omega,\tilde{{\bf U}}_{\omega})$ is the unique root one and then establish the proof. Suppose that subnetwork $B(\tilde{{\bf G}}_\omega,\tilde{{\bf F}}_\omega,\tilde{{\bf U}}_\omega)$ is controllable and subnetworks $B(\tilde{{\bf G}}_j,\tilde{{\bf F}}_j,\tilde{{\bf U}}_{j})$, $j\in[1,\omega-1]$, are structurally controllable. Additionally, we suppose that the length of the unique path from each generator to the root subnetwork is the same as $l$; otherwise, one can add the virtual nodes as in the proof of Theorem \ref{theorem-structuralcontrollability}.

Since $\mathbf{N}=\mathbf{N}_1\cup\mathbf{N}_2\cup\cdots\cup\mathbf{N}_\omega$ is a single-source-channel aggregation, the subnetwork, denoted by $\bigcup_{i=1}^{\omega-1}B(\tilde{{\bf G}}_i,\tilde{{\bf F}}_i,\tilde{{\bf U}}_i)$ is structurally controllable according to Theorem \ref{theorem-structuralcontrollability}, that is, $\Sigma_{i=1}^{\omega-1}{\bf N}_i$ is a group of disjoint root in-trees.

Given any two initial states $\alpha,\beta\in \mathscr{D}^n$, since we have set that the length of path from each generator to root subnetwork is the same, there holds $\mathbf{N}_\omega\cap\mathscr{U}=\emptyset$. Thus, $\mathscr{U}_\omega=\emptyset$ and the values of $\mathscr{U}_{\omega}$ between time instants $t=0$ and $t=l-1$, denoted by ${\bf u}^{\omega,0},{\bf u}^{\omega,2},\cdots,{\bf u}^{\omega,l-1}$, are actually determined by initial state $\alpha$ of the entire \texttt{BCN}. Once the input sequence ${\bf u}^{\omega,0},{\bf u}^{\omega,2},\cdots,{\bf u}^{\omega,l-1}$ has been given, the state of subnetwork $B(\tilde{{\bf G}}_\omega,\tilde{{\bf F}}_\omega,\tilde{{\bf U}}_\omega)$ at time instant $l$ can be determined as ${\bf x}^{\omega,l}$. Then, because of the controllability of subnetwork $B(\tilde{{\bf G}}_\omega,\tilde{{\bf F}}_\omega,\tilde{{\bf U}}_\omega)$, we can find another input sequence ${\bf u}^{\omega,l},{\bf u}^{\omega,l+1},\cdots,{\bf u}^{\omega,l+\kappa-1}$ to steer the state of this \texttt{BCN} from ${\bf x}^{\omega,l}$ to $\beta_{{\bf N}_{\omega}\cap\mathscr{X}}$. By Corollary \ref{cor-trajectory-controllability}, the root nodes in subnetwork $\bigcup_{i=1}^{\omega-1}B(\tilde{{\bf G}}_i,\tilde{{\bf F}}_i,\tilde{{\bf U}}_i)$ can be trajectory controllable w.r.t. sequence ${\bf u}^{\omega,l},{\bf u}^{\omega,l+1},\cdots,{\bf u}^{\omega,l+\kappa-1}$. The corresponding input sequence for generators is denoted as ${\bf u}^{0},{\bf u}^{1},\cdots,{\bf u}^{\kappa-1}$.

On the other hand, we can also determine the state $\alpha'_{\bigcup_{i=1}^{\omega-1} {\bf N}_i\cap{\mathscr{X}_i}}$ of subnetwork $\bigcup_{i=1}^{\omega-1}B(\tilde{{\bf G}}_i,\tilde{{\bf F}}_i,\tilde{{\bf U}}_i)$ via the input sequence ${\bf u}^{0},{\bf u}^{1},\cdots,{\bf u}^{\kappa-1}$. Since \texttt{BCN} $\bigcup_{i=1}^{\omega-1}B(\tilde{{\bf G}}_i,\tilde{{\bf F}}_i,\tilde{{\bf U}}_i)$ is structurally controllable, one can find another input sequence ${\bf u}^{\kappa},{\bf u}^{\kappa+1},\cdots,{\bf u}^{\kappa+l-1}$ to dominate its trajectory from $\alpha'_{\bigcup_{i=1}^{\omega-1} {\bf N}_i\cap{\mathscr{X}_i}}$ to $\beta_{\bigcup_{i=1}^{\omega-1} {\bf N}_i\cap{\mathscr{X}_i}}$.

By injecting the values ${\bf u}^{0},{\bf u}^{1},\cdots,{\bf u}^{\kappa-1},{\bf u}^{\kappa},{\bf u}^{\kappa+1},\cdots,{\bf u}^{\kappa+l-1}$ to generators ${\bm u}_i$ in turn, one can realize the reachability of \texttt{BCN} $B(\tilde{{\bf G}},\tilde{{\bf F}},{\bf U})$ from state $\alpha\in \mathscr{D}^n$ to $\beta\in \mathscr{D}^n$. Due to the arbitrariness of $\alpha,\beta\in\mathscr{D}$, the controllability of the entire \texttt{BCN} can be proved.
\end{proof}

\begin{theorem}\label{thm-aggregation-structuralcontrollability}
Given \texttt{BCN} $B(\tilde{{\bf G}},\tilde{{\bf F}},{\bf U})$ with acyclic digraph $\tilde{{\bf G}}$ and a single-source-channel network aggregation $\mathbf{N}=\mathbf{N}_1\cup\mathbf{N}_2\cup\cdots\cup\mathbf{N}_\omega$, the entire \texttt{BCN} is structurally controllable if and only if each subnetwork is structurally controllable.
\end{theorem}
\begin{proof}
On the one hand, in the single-source-channel network aggregation $\mathbf{N}=\mathbf{N}_1\cup\mathbf{N}_2\cup\cdots\cup\mathbf{N}_\omega$, if each subnetwork is composed of a series of root in-trees with leaves being generators, then the entire \texttt{BCN} still satisfies the condition 3) in Theorem \ref{theorem-structuralcontrollability}. One can imply that the entire \texttt{BCN} is structurally controllable.

On the other hand, according to Definition \ref{def-aggregation}, the network structure of each subnetwork must be a group of root in-trees. Thus, each subnetwork $B(\tilde{{\bf G}}_i,\tilde{{\bf F}}_i,\tilde{{\bf U}}_i)$ is structurally controllable.
\end{proof}


\subsection{Minimum Node Control Theorem}
Given \texttt{BN} $B({\bf G},{\bf F})$, since the in-neighbors of each simple node can only be channels, the minimal number of controlled nodes is obviously not less than $\omega^\ast-1$, where $\omega^\ast$ is the maximum vertex out-degree of ${\bf G}$. Although this bound is tight in some special cases, it is actually far from being sufficient to compute the minimal number of controlled nodes. Therefore, in the following, we shall utilize the network aggregation approach to provide an accurate solution for Problem \ref{problem-minimalnodecontrol}. Before proceeding, we analyze what kind of nodes is controlled necessarily or unnecessarily.
\begin{definition}
A digraph ${\bf G}$ is said to be ``spanned'' by an oriented path at vertex $v$ if, it becomes an oriented path via removing all the oriented edges jointing vertex $v$.
\end{definition}

Then, we present a scene when controlling one node suffices.
\begin{theorem}\label{thm-controllability-one-node}
Given \texttt{BN} $B({\bf G},{\bf F})$, the minimal number of controlled node is one if and only if
\begin{itemize}
  \item[1)] digraph ${\bf G}$ is an oriented path; or
  \item[2)] digraph ${\bf G}$ is a directed graph ``spanned'' by an oriented path at its beginning vertex.
\end{itemize}
\end{theorem}
\begin{proof}
The proof of sufficiency is trivial, since the structural controllability can be guaranteed if we let $\Lambda^\ast$ only contain the unique beginning vertex.

We now prove the necessity of this theorem. If the number of controlled nodes for \texttt{BN} $B({\bf G},{\bf F})$ is exactly one, then all of nodes must be accessible from one generator. More precisely, there is a path from this controlled node and traversing all vertices of ${\bf G}$. Thus, for digraph ${\bf G}$, the removed oriented edges (if exists) must end with this controlled node. It establishes the proof.
\end{proof}

\begin{theorem}\label{thm-ctr-reduction}
Given \texttt{BN} $B({\bf G},{\bf F})$ with digraph ${\bf G}=({\bf V},{\bf E})$.
\begin{itemize}
  \item[1)] The vertices with in-degree zero must be controlled.

  \item[2)] There must exist a minimum set of controlled nodes that does not contain the vertices on the strict path\footnote{A path $v_{k_1}\rightharpoonup v_{k_2} \rightharpoonup \cdots \rightharpoonup v_{k_p}$ is said to be strict if, $v_{k_i}$ is the unique out-neighbor of $v_{k_{i-1}}$ and $v_{k_{i-1}}$ is the unique in-neighbor of $v_{k_i}$ for $i\geq2$.} $\Omega$, except for the beginning vertex;

  \item[3)] If there is a vertex $v$ with two disjoint out-neighbors $v'$ and $v''$, where $v$ is the unique in-neighbor of $v'$ and $v''$, then it is only needed to control one of vertices $v'$ and $v''$.

  \item[4)] If there is a vertex $v$ with two disjoint out-neighbors $v'$ and $v''$, where $v$ is the unique in-neighbor of $v'$ and $v'$ and does not lie in any cycle, then there must exist a minimum set of controlled nodes that does not contain vertex $v'$.

  \item[5)] If $(v_o,v_o)\in{\bf E}$, then vertex $v_o$ must be controlled.
\end{itemize}
\end{theorem}
\begin{proof}
$1)$: If vertex $v$ has no in-degree, then it is unaccessible from any generator. Therefore, it must be controlled.

$2)$: If there is a strict path $P$ containing edge $(v_{\mu},v_{o})$ with $v_o\in\Lambda^\ast$, then the out-degree of node $v_\mu$ must be zero in \texttt{BCN} $B(\tilde{{\bf G}},\tilde{{\bf F}},\Lambda^\ast)$. So, if we control the node set $\left(\Lambda^\ast\backslash \{v_o\}\right)\cup\{v_\mu\}$ whose cardinality is equal to that of $\Lambda^\ast$, then the \texttt{BCN} is still structurally controllable.

$3)$: Since the out-neighbors of $v$ are $v'$ and $v''$, $v$ is not a channel. It is claimed that at least one of $v'$ and $v''$ needs to be controlled. Besides, once $v'$ (or $v''$) is controlled, $v''$ (or $v'$) lies on a strict path $v\rightarrow v''$ (or $v \rightarrow v'$). According to proposition 2) in Theorem \ref{thm-ctr-reduction}, $v''$ (or $v'$) does not need to be controlled.

The proof of $4)$ and $5)$ can be directly obtained by resorting to case 4) in Theorem \ref{thm-ctr-reduction} and Theorem \ref{thm-controllability-one-node}, respectively.
\end{proof}

\begin{theorem}\label{thm-aggregation-controllable}
Given \texttt{BN} $B({\bf G},{\bf F})$ and a single-source-channel network aggregation $\mathbf{N}=\mathbf{N}_1\cup\mathbf{N}_2\cup\cdots\cup\mathbf{N}_\omega$. Denote the strongly connected components of ${\bf G}$ by ${\bf S}_1$, ${\bf S}_2$, $\cdots$, ${\bf S}_\kappa$. For each subnetwork $B({\bf G}_i,{\bf F}_i, {\bf U}_i)$, denote its minimal set of controlled nodes by $\Lambda_i^\ast$ whose cardinality is $N_i^\ast$. Then the minimal number of controlled nodes for the entire \texttt{BN} $B({\bf G},{\bf F})$ is $\Lambda^\ast=\bigcup_{i=1}^{\omega}\Lambda_i^\ast$ and the number of controlled nodes is $N^\ast=\Sigma_{i=1}^{\omega}N_i^\ast$ if there are not two different ${\bf N}_i$ and ${\bf N}_j$ satisfying that ${\bf N}_i\cap {\bf S}_k\neq\emptyset$ and ${\bf N}_j\cap {\bf S}_k\neq\emptyset$ for any $k\in[1,\kappa]$.
\end{theorem}
\begin{proof}
According to Theorem \ref{theorem-structuralcontrollability}, $\Lambda^\ast$ is a feasible set of controlled nodes to force the structural controllability of \texttt{BN}. Reversely, since this acyclic aggregation is single-source-channel, for each directed edge $e_{ij}$ with $v_i\in{\bf N}_i$ and $v_j\in{\bf N}_j$, vertex $v_j$ does not need to be controlled based on the case 3) of Theorem \ref{thm-ctr-reduction}. Hence, by the necessity of Theorem
\ref{thm-aggregation-structuralcontrollability}, one can conclude that the set $\Lambda^\ast$ must be minimal.
\end{proof}

Next, given \texttt{BN} $B({\bf G},{\bf F})$, on the basis of the single-source-channel aggregation in Theorem \ref{thm-aggregation-controllable}, we can provide an efficient and explicit procedure to solve Problem \ref{problem-minimalnodecontrol} as well as all the feasible solutions.

Consider subnetwork $B({\bf G}_k,{\bf F}_k,{\bf U}_k)$. Let set ${\bf N}_k \cap \mathscr{X}=\{ v_{x_1^k},v_{x_2^k},\cdots,v_{x_{n_k}^k}\}$ with $x_1^k < x_2^k < \cdots < x_{n_k}^k$, then we denote the set of controlled nodes by set $S_k$. Equivalently, it can be defined as an $n_k$-dimensional binary vector $V_{S_k}=[\varsigma^k_1,\varsigma^k_2,\cdots,\varsigma^k_{n_k}]$, where $\varsigma^k_i=1$ if $x^k_i\in S_k$; $\varsigma^k_i=0$, otherwise. Via the correspondence ``$\mapsto$'' defined in Section \ref{sec-preliminaries}, one has that $\varsigma_i^k \mapsto \zeta_i^k:=\delta_2^{2-\varsigma_i^k}$ for $i\in[1,n_k]$. The adjacency matrix for state nodes can be defined by a Boolean matrix $A_k=(a^k_{ij})_{n_k \times n_k}$.

Suppose that the subnetwork $\Sigma({\bf G}_i,{\bf F}_i, {\bf U}_i)$ is structurally controllable through controlling the vertices in $S_k$. Such subnetwork should be restrained by the following three conditions.
\begin{itemize}
  \item[$C1)$] If $v_{x_i^k}$ has been a controlled node, then it derives $\varsigma^k_i=0$; and if its in-degree is zero, then $\varsigma^k_i=1$;

  \item[$C2)$] The out-degree of each state node is no more than $1$, that is,
  $$ \bigvee\limits_{i=1}^{n_k}\left[\sum_{j=1}^{n_k} a^k_{ij}(1-\varsigma^k_j)\right]\leq 1,$$
  where $a \bigvee b:=\max\{a,b\}$ for $a,b\in\mathbb{N}$.

  \item[$C3)$] The network structure after control is acyclic. By depth-first search algorithm, one can compute all the cycles in the digraph ${\bf G}_k$, denoted as $C_{k,1}$,
  $C_{k,2}$, $\cdots$, $C_{k,\kappa_k}$, where $C_{k,i}$ can be written as a vector $\gamma_{k,i}=[\gamma_{k,i}^1,\gamma_{k,i}^2,\cdots,\gamma_{k,i}^{n_k}]\in \mathscr{D}^{n_k}$ for $i\in[1,\kappa_k]$, where $v_{x^j_k}\in C_{k,i}$ implies that $\gamma_{k,i}^{j}=1$. Then, one has that
  $$\sum_{j=1}^{n_k}\gamma^j_{k,i}  \varsigma^k_j>0.$$
\end{itemize}

Let $W^k_1$ and $W^k_2$ respectively denote the nodes that have to be controlled and those need not be controlled in \texttt{BCN} $B({\bf G}_k,{\bf F}_k,{\bf U}_k)$. The solution to Problem \ref{problem-minimalnodecontrol} is equivalent to determining the solution of the following optimal problem:
\begin{equation}\label{equ-optimal}
\begin{aligned}
&\min ~\eta_k=\sum\limits_{j=1}^{n_k}\varsigma^k_{j},~k=1,2,\cdots,\kappa,\\
&\text{s.t}\left\{\begin{aligned} &\varsigma^k_i=0, ~v_i\in W^k_1,\\ &\varsigma^k_i=1, ~v_i\in W^k_2, \\&\bigvee\limits_{i=1}^{n_k}\left[\sum_{j=1}^{n_k} a^k_{ij}(1-\varsigma^k_{j})\right]\leq 1,\\ &\sum_{j=1}^{n_k}\gamma^j_{k,i}  \varsigma^k_j>0. \end{aligned}\right.
\end{aligned}
\end{equation}

To this end, the \texttt{STP} of matrices is utilized to address the above optimal problem.
\begin{theorem}\label{thm-solveoptimal}
Given \texttt{BN} $B({\bf G},{\bf F})$ and a single-source-channel aggregation $\mathbf{N}=\mathbf{N}_1\cup\mathbf{N}_2\cup\cdots\cup\mathbf{N}_\omega$ satisfying that there are not two different ${\bf N}_i$ and ${\bf N}_j$ such that ${\bf N}_i\cap {\bf S}_k\neq\emptyset$ and ${\bf N}_j\cap {\bf S}_k\neq\emptyset$ for any $k\in[1,\kappa]$. For every subnetwork $B({\bf G}_k,{\bf F}_k,{\bf U}_k)$, the set $\Lambda^\ast_k \mapsto \delta_{2^{n_k}}^{s_k}$ is a feasible solution to Problem \ref{problem-minimalnodecontrol}, if and only if, $J_1 \ltimes \bar{M}^k \ltimes \delta_{2^{n_k}}^{s_k} \in\mathscr{D}$ and
$J_1 \ltimes \tilde{M}^k \ltimes \delta_{2^{n_k}}^{s_k} \neq 0$, where $J_1=[1,0]$.
\end{theorem}
\begin{proof}
Consider the condition $C2)$. We have that
\begin{equation}\label{equation-1}
\begin{aligned}
\zeta^k_j&=\Psi^{n_k-1} \ltimes (\zeta^k_{j+1}\zeta^k_{j+2} \cdots \zeta^k_{n}) \ltimes (\zeta^k_1\zeta^k_2 \cdots \zeta^k_{j-1}) \ltimes \zeta^k_j\\
&=\Psi^{n_k-1} W_{[2^j,2^{n_k-j}]} \ltimes (\zeta^k_1\zeta^k_2 \cdots \zeta^k_{j-1})\ltimes \zeta^k_j \ltimes (\zeta^k_{j+1}\zeta^k_{j+2} \cdots \zeta^k_{n})\\
&=\Psi^{n_k-1} W_{[2^j,2^{n_k-j}]} \ltimes \zeta^k = M^k_j \ltimes \zeta^k,
\end{aligned}
\end{equation}
where $M_j^k=\Psi^{n_k-1} W_{[2^j,2^{n_k-j}]}$ and $\zeta^k=\ltimes_{j=1}^{k_{n}}\zeta^k_j$. It is claimed that
\begin{equation}
\sum_{j=1}^{n_k} a^k_{ij}(1-\varsigma^k_j)=\sum_{j=1}^{n_k} a^k_{ij} J_1(\delta_2^1 \otimes {\bf 1}^\top_{2^{n_k}}- M^k_j) \ltimes \zeta^k= J_1 \ltimes \bar{M}^k_i \ltimes \zeta^k,
\end{equation}
where $\bar{M}_i^k=\sum_{j=1}^{n_k} a^k_{ij} (\delta_2^1 \otimes {\bf 1}^\top_{2^{n_k}}- M^k_j)$. Subsequently, we define the matrix $\bar{M}^k$ as
$$ [\bar{M}^k]_{\mu\nu}= \max\limits_{i\in[1,n_k]}{[\bar{M}^k_i]_{\mu\nu}}.$$
Condition $C2)$ holds by controlling the node set $V_{S_k}$ corresponding to canonical vector $\zeta_k=\delta_{2^n}^{s_k}$, that is, the $s_k$-th component of $\bar{M}_k$ is $1$ or $0$, or alternatively, belongs to $\mathscr{D}$ and vice versa.

Besides condition $C3)$ derives that
\begin{equation}
\sum_{j=1}^{n_k}\gamma^j_{k,i}  \varsigma^k_j = \sum_{j=1}^{n_k}\gamma^j_{k,i} J_1 M_j^k \zeta^k = J_1 \sum_{j=1}^{n_k}\gamma^j_{k,i} M_j^k \zeta^k = J_1 \ltimes \tilde{M}_i^k \ltimes \zeta^k,
\end{equation}
where $\tilde{M}_i^k = \sum_{j=1}^{n_k}\gamma^j_{k,i} M_j^k$. Accordingly, the condition $3)$ can imply the $J_1 \ltimes \tilde{M}^k \ltimes \delta_{2^{n_k}}^{s_k} =1$, and vice versa.
\end{proof}

\begin{theorem}
Given \texttt{BN} $B({\bf G},{\bf F})$ and a single-source-channel aggregation $\mathbf{N}=\mathbf{N}_1\cup\mathbf{N}_2\cup\cdots\cup\mathbf{N}_\omega$ satisfying that there are not two different ${\bf N}_i$ and ${\bf N}_j$ such that ${\bf N}_i\cap {\bf S}_k\neq\emptyset$ and ${\bf N}_j\cap {\bf S}_k\neq\emptyset$ for any $k\in[1,\kappa]$. Set $\Lambda^\ast$ is the minimum controlled node set w.r.t. Problem \ref{problem-minimalnodecontrol}, if and only if, $\Lambda^\ast= \bigcup_{k=1}^{\omega}\Lambda^\ast_k$, where $J_1 \ltimes \bar{M}^k \ltimes \delta_{2^{n_k}}^{s_k} \in\mathscr{D}$ and
$J_1 \ltimes \tilde{M}^k \ltimes \delta_{2^{n_k}}^{s_k} \neq 0$.
\end{theorem}

\section{Distributed Pinning Controllability of BNs}\label{sec-pin-controllability}
In Section \ref{sec-structuralcontrollability}, Theorem \ref{theorem-structuralcontrollability} has been established in the case without available node dynamics. However, while the node dynamics can be identifiable, it is still potential to design the pinning controller to make an uncontrollable \texttt{BN} controllable. From the viewpoint of network structure, several disadvantages of the traditional pinning controllability strategies, such as problems ({\bf P1}) and ({\bf P4}) in \cite{chenhw2016scis,liff2020cyber,liff2020tcns1523,lujq2016ieeetac1658,liuzq2020tcns}, can be reduced to some extent. In \cite{zhongjie2019new}, Zhong {\em et al.} designed the distributed pinning controller for the stabilization of \texttt{BNs} for the first time, based on the network structure. However, the acyclic network structure is only applicable for the global stabilization.

In the following, we would like to present a novel pinning strategy for the sake of making \texttt{BCN} $B(\tilde{{\bf G}},\tilde{{\bf F}},{\bf U})$ controllable. Denote by $\Gamma$ the set of pinning nodes. For $k\in\Gamma$, the inputs are injected in the form of
\begin{equation}\label{equation-pinning-bn}
{\bm x}_k(t+1)=\tilde{{\bm u}}_k(t)\tilde{\oplus}_k \tilde{{\bm f}}_k\big([{\bm x}_i(t)]_{i\in\tilde{\mathbf{X}}_k}\big),
\end{equation}
where $\oplus_k:\mathscr{D}^2\rightarrow\mathscr{D}$ is a binary logical operator, and the injected input $\tilde{{\bm u}}_k(t)$ can be open-loop control or feedback control. The dynamics of functions $\tilde{{\bm f}}_k$, $k\in[1,n]$, have been given in \texttt{BCN} $B(\tilde{{\bf G}},\tilde{{\bf F}},{\bf U})$. Besides, for $k\in \tilde{{\bf V}}\backslash\Gamma$, its dynamics is identical with that in $B(\tilde{{\bf G}},\tilde{{\bf F}},{\bf U})$. The portions of (\ref{equation-pinning-bn}) that need to be designed are the pinning node set $\Gamma$, logical operator $\tilde{\oplus}_k$ and the control input $\tilde{{\bm u}}_k(t)$. The research idea of this part can be described as Fig. \ref{fig-idea}.
\begin{figure}[h!]
\centering
\includegraphics[width=0.48\textwidth=0.48]{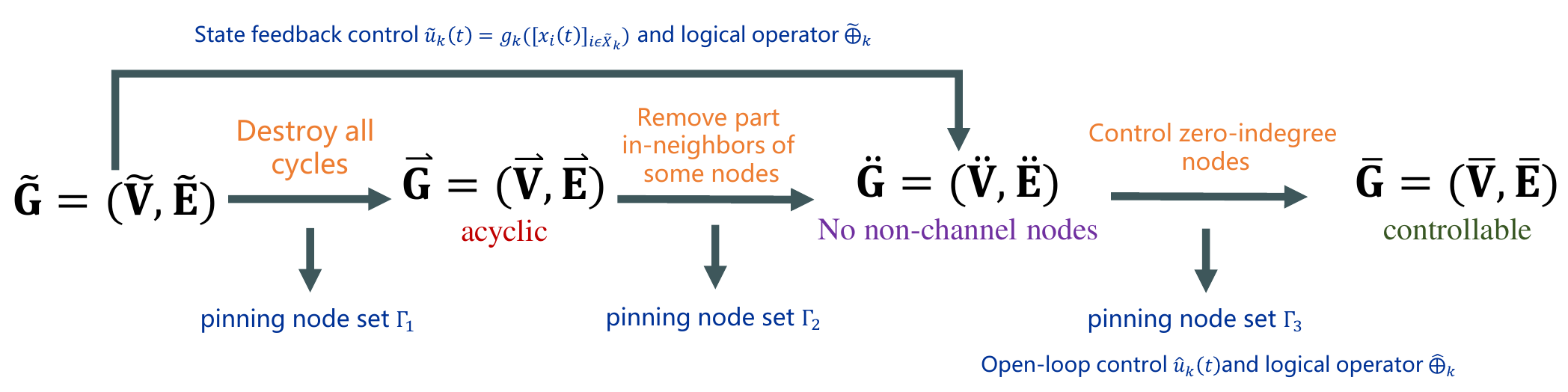}
\caption{The procedure of pinning control design in this section.}\label{fig-idea}
\end{figure}

\subsection{Selection of the Pinning Node Set}
Motivated by condition $3)$ in Theorem \ref{theorem-structuralcontrollability}, the candidate pinning nodes are possibly the nodes that do not satisfy condition $({\bf C1})$ or lie in a cycle. Then inputs $\tilde{{\bm u}}_k(t)$, $k\in\Gamma$, are designed to modify the adjacency relationship of pinning nodes, aiming to satisfy condition 3) in Theorem \ref{theorem-structuralcontrollability}. In the design of pinning controller, it is crucial to develop a selection procedure for the pinning nodes, but it is pretty difficult by using traditional \texttt{ASSR} approaches \cite{chenhw2016scis,liff2020cyber,liff2020tcns1523,lujq2016ieeetac1658}.

First of all, in diagraph $\tilde{{\bf G}}$, by the depth-first search algorithm, one can find all the cycles in $\tilde{{\bf G}}$ including self loops, which are denoted as $\vec{C}_1,\vec{C}_2,\cdots,\vec{C}_s$. For each cycle $\vec{C}_i$, $i\in[1,s]$, we can arbitrarily select an oriented edge $e_{\sigma_i\varsigma_i} \in \vec{C}_i$. Then, the acyclic digraph $\vec{{\bf G}}:=(\vec{{\bf V}},\vec{{\bf E}})$ can be constructed as $\vec{{\bf V}}:=\tilde{{\bf V}}$ and $\vec{{\bf E}}:=\tilde{{\bf E}}\backslash \bigcup_{i=1}^s\left\{e_{\sigma_i\varsigma_i}\right\}$. The above process produces an
acyclic digraph $\vec{{\bf G}}$, and the first type of pinning node set can be determined as $\Gamma_1:=\bigcup_{i=1}^s\{\varsigma_i \mid i\in[1,s]\}$.

Secondly, consider the above acyclic digraph $\vec{{\bf G}}$. Let $S=\{\varrho_1,\varrho_2,\cdots,\varrho_\iota\}\subseteq [m+1,n]$ be the index set of all state nodes in $\vec{{\bf G}}$ satisfying $\mid\vec{\mathbf{N}}^{\text{out}}_{\varrho_i}\mid \geq 2$ with $i\in[1,\iota]$. Furthermore, in order to control as fewer nodes as possible, we define a mapping $\odot:\bigcup_{i=1}^{\iota}\vec{\mathbf{N}}^{\text{out}}_{\varrho_i}\rightarrow\mathbb{R}$ as
\begin{equation}\label{equ-odot}
\odot(v_j)=\sum_{i=1}^{\iota}{\bf 1}_{\left\{v_j\in \vec{\mathbf{N}}^{\text{out}}_{\varrho_i}\right\}},~v_j\in\bigcup_{i=1}^{\iota} \vec{{\bf N}}^{\text{out}}_{\varrho_i},
\end{equation}
where ${\bf 1}_{\{\cdot\}}$ is the indicator function. Therefore, for each $\varrho_i\in S$, we can define
$$\kappa_{\varrho_i}:=\text{argmin}_{v_j\in \vec{\mathbf{N}}^{\text{out}}_{\varrho_i}}\left\{ \odot(v_j)\right\}.$$
Consequently, the second type of pinning nodes can be selected as $\Gamma_2:=\bigcup_{i=1}^{\iota}\{ j \mid v_j\in\vec{\mathbf{N}}^{\text{out}}_{\varrho_i}\}\backslash \{\kappa_{\varrho_i}\}$, and we can construct the network structure as $\ddot{{\bf G}}=(\ddot{{\bf V}},\ddot{{\bf E}})$, where $\ddot{{\bf V}}:={\bf V}$ and $\ddot{{\bf E}}:=\vec{{\bf E}}\backslash \bigcup_{i=1}^{\iota} \{e_{\varrho_i,j}\mid j\in\vec{\mathbf{N}}^{\text{out}}_{\varrho_i}\backslash \{v_{\kappa_{\varrho_i}}\}\}$. In digraph $\ddot{{\bf G}}$, the in-neighbor set (respectively, out-neighbor set) of node $v_i$ is denoted as $\ddot{\mathbf{N}}^{\text{in}}_i$ (respectively, $\ddot{\mathbf{N}}^{\text{out}}_i$). In this way, the obtained digraph $\ddot{{\bf G}}$ is acyclic, and the in-neighbor set of each state node is either empty or full with channels.

Thirdly, the state nodes with in-degree zero in $\ddot{{\bf G}}$ must be controlled, thus set $\Gamma_3$ is composed of state nodes with in-degree zero in $\ddot{{\bf G}}$. In conclusion, the pinning node set can be selected as $\Gamma:=\Gamma_1\cup\Gamma_2\cup\Gamma_3$.

\subsection{Control Design for Pinning Nodes}\label{subsection-compute}
In this subsection, we design the logical coupling $\tilde{\oplus}_k$ and control inputs $\tilde{{\bm u}}_k(t)$ for different types of nodes. Denote the canonical form of ${\bm x}_k$ and $\tilde{{\bm u}}_k$ in (\ref{equation-pinning-bn}) respectively as $x_k:=\delta_2^{2-{\bm x}_k}$ and $\tilde{u}_k:=\delta_2^{2-\tilde{{\bm u}}_k}$.

In the first part, we first consider pinning nodes in $\Gamma_1$ and $\Gamma_2$: For $k\in\Gamma_1\cup\Gamma_2$, input $\tilde{{\bm u}}_k(t)$ is designed as the feedback of variables ${\bm x}_i(t)$, $i\in\tilde{\mathbf{X}}_k$, that is, $\tilde{{\bm u}}_k(t)={\bm g}_k([{\bm x}_i(t)]_{i\in\tilde{\mathbf{X}}_k})$. Without loss of generality, we assign ${\bf N}^{\text{in}}_i=\{v_{\alpha_i^1},v_{\alpha_i^2},\cdots,v_{\alpha_i^{\mid {\bf N}^{\text{in}}_i \mid}}\}$ and $\ddot{{\bf N}}^{\text{in}}_i=\{v_{\alpha_i^{\kappa_1}},v_{\alpha_i^{\kappa_2}},\cdots,v_{\alpha_i^{\kappa_{\mid \ddot{{\bf N}}^{\text{in}}_i \mid}}}\}$ with $\alpha_i^1<\alpha_i^2<\cdots<\alpha_i^{\mid {\bf N}^{\text{in}}_i \mid}$ and $\kappa_1<\kappa_2<\cdots<\kappa_{\mid \ddot{{\bf N}}^{\text{in}}_i \mid}$. In this setting, it holds that
\begin{equation*}
\begin{aligned}
\ltimes_{k=1}^{\mid {\bf N}^{\text{in}}_i \mid} x_{\alpha_i^k} &= W_{[2,2^{\kappa_{\mid \ddot{{\bf N}}^{\text{in}}_i \mid}-1}]} x_{\alpha_i^{\kappa_{\mid \ddot{{\bf N}}^{\text{in}}_i \mid}}} (\ltimes_{k=1}^{\kappa_{\mid {\bf N}^{\text{in}}_i \mid}-1} x_{\alpha_i^k}) (\ltimes_{k=\kappa_{\mid \ddot{{\bf N}}^{\text{in}}_i \mid}-1}^{\mid {\bf N}^{\text{in}}_i \mid} x_{\alpha_i^k})\\
&=W_{i,{\bf N}^{\text{in}}_i,\ddot{{\bf N}}^{\text{in}}_i} \ltimes (\ltimes_{k=1}^{\mid \ddot{{\bf N}}^{\text{in}}_i \mid} x_{\alpha_i^{\kappa_k}}) (\ltimes_{k\in \bar{\mathbf{N}}^{\text{in}}_i\backslash \ddot{{\bf N}}_i^{\text{in}}}x_{\alpha_i^{\kappa_k}}),
\end{aligned}
\end{equation*}
where $W_{i,{\bf N}^{\text{in}}_i,\ddot{{\bf N}}^{\text{in}}_i}=\ltimes_{k={\mid \ddot{{\bf N}}^{\text{in}}_i \mid}}^{1} W_{[2,2^{\kappa_k+\mid \ddot{{\bf N}}^{\text{in}}_i \mid-k-1}]}$. Then, one can find a matrix $F_k\in\mathscr{L}_{2\times 2^{\mid \mathbf{N}_{k}^{\text{in}} \mid}}$ such that
\begin{equation}\label{equation-equa}
F_k=A_k\ltimes\left(I_{2^{\mid \ddot{\mathbf{N}}^{\text{in}}_k \mid}}\otimes {\bf 1}_{2^{\mid \mathbf{N}^{\text{in}}_k \backslash \ddot{\mathbf{N}}^{\text{in}}_k \mid}}^\top\right) \ltimes W^\top_{k,{\bf N}^{\text{in}}_k,\ddot{{\bf N}}^{\text{in}}_k}, ~k\in \Gamma_1\cup\Gamma_2
\end{equation}
with matrix $A_k\in\mathscr{L}_{2\times 2^{\mid \ddot{\mathbf{N}}^{\text{in}}_k\mid}}$ satisfying
\begin{equation}\label{equation-equa2}
\text{Blk}_1\left(A_k W_{[2,2^{j-1}]}\right)\not=\text{Blk}_2\left(A_k W_{[2,2^{j-1}]}\right)
\end{equation}
for all $j\in \left[1,\mid \ddot{\mathbf{N}}^{\text{in}}_k\mid\right]$.

Hence, by Lemma \ref{lemma-structurematrix}, the multi-linear form of node dynamics (\ref{equation-pinning-bn}) with feedback control $\tilde{{\bm u}}_k(t)={\bm g}_{k}([{\bm x}_i(t)]_{i\in \tilde{{\bf X}}_k})$ can be derived as
\begin{equation*}
\begin{aligned}
x_k(t+1)&=L_{\oplus_k}L_{{\bm g}_k}\left(\ltimes_{i\in\tilde{\mathbf{X}}_{k}} x_i(t)\right) L_{\tilde{{\bm f}}_k}\left(\ltimes_{i\in\tilde{\mathbf{X}}_{k}} x_i(t)\right)\\
&=L_{\oplus_k}L_{{\bm g}_k}\left(I_{2^{\mid \tilde{\mathbf{X}}_{k}\mid}}\otimes L_{\tilde{{\bm f}}_k}\right)\Phi_{2^{\mid \tilde{\mathbf{X}}_{k} \mid}}\left(\ltimes_{i\in\tilde{\mathbf{X}}_{k}} x_i(t)\right),
\end{aligned}
\end{equation*}
where $L_{\oplus_k}\in\mathscr{L}_{2\times 4}$, $L_{{\bm g}_k}\in\mathscr{L}_{2\times 2^{\mid {\bf N}^{\text{in}}_k \mid}}$ and $L_{\tilde{{\bm f}}_k}\in\mathscr{L}_{2\times 2^{\mid {\bf N}^{\text{in}}_k \mid}}$ are respectively the structure matrices of logical functions $\oplus_k$, ${\bm g}_k$ and $\tilde{{\bm f}}_k$. Hence, matrices $L_{\oplus_k}$ and $L_{{\bm g}_k}$ can be solved from the following equations
\begin{equation}\label{equ-logicalequation}
L_{\oplus_k}L_{{\bm g}_k}\left(I_{2^{\mid \tilde{\mathbf{X}}_{k}\mid}}\otimes L_{\tilde{{\bm f}}_k}\right)\Phi_{2^{\mid \tilde{\mathbf{X}}_{k}\mid }}=F_k.
\end{equation}
Fortunately, equation (\ref{equ-logicalequation}) is inevitably solvable as proved in \cite{Liff2016TNNLS1585}.

The second part considers the pinning nodes in $\Gamma_3$: For $k\in\Gamma_3$, the input $\tilde{{\bm u}}_k$ is in the open-loop form and the logical operator $\tilde{\oplus}_k$ is directly designed as $\wedge$.

To sum up, the pinning controller can be designed as (\ref{equation-final}), since there might exist some common elements between $\Gamma_1\cup\Gamma_2$ and $\Gamma_3$.
\begin{figure*}[!ht]
\centering
\begin{equation}\label{equation-final}
\left\{\begin{array}{ll}
{\bm x}_k(t+1)={\bm g}_k\left([{\bm x}_i(t)]_{i\in\tilde{\mathbf{X}}_k}\right)\tilde{\oplus}_k \tilde{{\bm f}}_k\left([{\bm x}_i(t)]_{i\in\tilde{\mathbf{X}}_k}\right),&k\in (\Gamma_1\cap\Gamma_2)\backslash\Gamma_3,\\
{\bm x}_k(t+1)=\tilde{{\bm u}}_k(t)\wedge \tilde{{\bm f}}_k\left([{\bm x}_i(t)]_{i\in\tilde{\mathbf{X}}_k}\right),&k\in \Gamma_3\backslash(\Gamma_1\cap\Gamma_2),\\
{\bm x}_k(t+1)=\tilde{{\bm u}}_k(t) \wedge \left({\bm g}_k\left([{\bm x}_i(t)]_{i\in\tilde{\mathbf{X}}_k}\right)\tilde{\oplus}_k {\bm f}_k\left([{\bm x}_i(t)]_{i\in\tilde{\mathbf{X}}_k}\right)\right),&k\in (\Gamma_1\cap\Gamma_2)\cap\Gamma_3,\\
{\bm x}_k(t+1)=\tilde{{\bm f}}_k\left([{\bm x}_i(t)]_{i\in\tilde{\mathbf{X}}_k}\right),&k\in [m+1,n]\backslash \Gamma,\\
{\bm x}_k(t+1)={\bm u}_k(t),&k\in [1,m].
\end{array}
\right.
\end{equation}
\hrulefill
\vspace*{1pt}
\end{figure*}

\begin{theorem}
Given an uncontrollable \texttt{BCN} $B(\tilde{{\bf G}},\tilde{{\bf F}},{\bf U})$, the pinning controlled \texttt{BCN} (\ref{equation-final}) is controllable.
\end{theorem}
\begin{proof}
To establish this theorem, we need to prove that the network structure of \texttt{BCN} (\ref{equation-final}) is acyclic and satisfies condition ({\bf C1}). Without considering the open-loop control input $\tilde{{\bm u}}_k$, we prove that the network structure of (\ref{equation-final}) is $\ddot{{\bf G}}$. Notice that the structure matrix for the dynamics of node $k$ is $F_k$, by plugging (\ref{equation-equa}) into its multi-linear form, one has that
\begin{equation}
\begin{aligned}
x_k&(t+1)=F_k (\ltimes_{i\in \tilde{{\bf X}}_k} x_i(t))\\
&= A_k(I_{2^{\mid \ddot{\mathbf{N}}^{\text{in}}_k \mid}}\otimes {\bf 1}_{2^{\mid \mathbf{N}^{\text{in}}_k\backslash\ddot{\mathbf{N}}^{\text{in}}_k\mid}}^\top) W^\top_{k,{\bf N}^{\text{in}}_k,\ddot{{\bf N}}^{\text{in}}_k} (\ltimes_{i\in \tilde{{\bf X}}_k} x_i(t))\\
&= A_k(I_{2^{\mid \ddot{\mathbf{N}}^{\text{in}}_k \mid}}\otimes {\bf 1}_{2^{\mid \mathbf{N}^{\text{in}}_k\backslash\ddot{\mathbf{N}}^{\text{in}}_k\mid}}^\top) W^\top_{k,{\bf N}^{\text{in}}_k,\ddot{{\bf N}}^{\text{in}}_k} W_{k,{\bf N}^{\text{in}}_k,\ddot{{\bf N}}^{\text{in}}_k} \\
&~~~~~~~~~~~~~~~~~~~~~~~~~~~~~~(\ltimes_{j=1}^{\mid \ddot{{\bf N}}^{\text{in}}_k \mid} x_{\alpha_k^{\kappa_j}}) (\ltimes_{j\in \mathbf{N}^{\text{in}}_k\backslash \ddot{\mathbf{N}}^{\text{in}}_k } x_{\alpha_k^{\kappa_j}})\\
&=A_k(I_{2^{\mid \ddot{\mathbf{N}}^{\text{in}}_k \mid}}\otimes {\bf 1}_{2^{\mid \mathbf{N}^{\text{in}}_k\backslash \ddot{{\bf N}}^{\text{in}}_k \mid}}^\top) (\ltimes_{j=1}^{\mid \ddot{{\bf N}}^{\text{in}}_k \mid} x_{\alpha_k^{\kappa_j}}) (\ltimes_{j\in \mathbf{N}^{\text{in}}_k \backslash \ddot{{\bf N}}^{\text{in}}_k}x_{\alpha_k^{\kappa_j}})\\
&=A_k(\ltimes_{j=1}^{\mid \ddot{{\bf N}}^{\text{in}}_k \mid} x_{\alpha_k^{\kappa_j}})
\end{aligned}
\end{equation}
where the fourth equation in the above holds because every swap matrix is also a permutation matrix \cite{chengdz2011springer}, and the establishment of the fifth equation is due to
$$(I_{2^{\mid \ddot{\mathbf{N}}^{\text{in}}_k \mid}}\otimes {\bf 1}_{2^{\mid \mathbf{N}^{\text{in}}_k\backslash \ddot{{\bf N}}^{\text{in}}_k \mid}}^\top) (\ltimes_{j=1}^{\mid \ddot{{\bf N}}^{\text{in}}_k \mid} x_{\alpha_k^{\kappa_j}})=(\ltimes_{j=1}^{\mid \ddot{{\bf N}}^{\text{in}}_k \mid} x_{\alpha_k^{\kappa_j}})\otimes {\bf 1}_{2^{\mid \mathbf{N}^{\text{in}}_k \backslash \ddot{\mathbf{N}}^{\text{in}}_k \mid}}^\top.$$

Moreover, since matrix $A_k$ satisfies the condition (\ref{equation-equa2}) and $\ltimes_{j=1}^n x_j=W_{[2,2^{j-1}]} x_j \ltimes_{i=1}^{j-1}x_i \ltimes_{i=j+1}^{n}x_j$, one can conclude that every variable ${\bm x}_j$ is functional.

Finally, the pinning nodes in $\Gamma_3$ guarantee that the in-neighbor set of any state vertex is nonempty. Therefore, \texttt{BCN} (\ref{equation-final}) is controllable.
\end{proof}

\subsection{Discussions and Comparisons}
Compared with the traditional \texttt{ASSR} approach as in \cite{chenhw2016scis,liff2020cyber,liff2020tcns1523,lujq2016ieeetac1658,liuzq2020tcns}, this novel pinning approach is equipped with the following four superiorities:
\begin{itemize}
  \item[1)] The time complexity of deriving the network structure $\tilde{{\bf G}}$ of \texttt{BCN} $B(\tilde{{\bf G}},\tilde{{\bf F}},{\bf U})$ is polynomial w.r.t. the node number $n+m$; it is bounded by $\Theta(n^2+mn+m)$. The pinning nodes are selected by utilizing the depth-first search on digraph $\tilde{{\bf G}}$, thus its time complexity is also bounded by $\Theta((n+m)^2)$. To solve all logical equations (\ref{equ-logicalequation}), all operators are determined by the structure matrix $L_{\tilde{{\bm f}}_i}\in\mathscr{L}_{2\times 2^{\mid\tilde{{\bf N}}^{\text{in}}_i\mid}}$, and this process can be implemented within time $\Theta(n2^{3d^\ast})$, where $d^\ast$ is the maximal in-degree of all vertices. To sum up, the time complexity of our approach is totally upper bounded by $\Theta(n2^{3d^\ast}+2(n+m)^2)$. The actual experiments have pointed out that ``the realistic biological networks are always sparsely connected'' \cite{jeong2000nature}, thus the number $d^\ast$ would not be pretty large. Therefore, this pinning strategy provides a way to force the controllability of the large-scale \texttt{BCNs}.

  \item[2)] Additionally, it is noted that feedback control ${\bm g}_k$ only depends on the functional variables of $\tilde{{\bm f}}_k$, thus it is of distributed form and is more concise than those in \cite{chenhw2016scis,liff2020cyber,liff2020tcns1523,lujq2016ieeetac1658,liuzq2020tcns}.

  \item[3)] Another important thing is that we develop a way to design the pinning controller so as to make \texttt{BCNs} controllable rather than just to check whether they are controllable under the pre-assigned pinning controller. This is the most essential difference.

  \item[4)] Finally, the pinning node set can be selected as $\Gamma$ in a polynomial time. By comparison, it is superior to the traditional methods, which may need to inject the control inputs on all the state nodes.
\end{itemize}

\subsection{Output Tracking/Synchronization/Regulation for BCNs}
Last but not least, our approach not only suits the pinning controllability, but can also be utilized to improve the methods of pinning output tracking/synchronization/regulation of \texttt{BCN} $B(\tilde{{\bf G}},\tilde{{\bf F}},{\bf U})$. All these problems require that the outputs can track to a given reference trajectory. Thus, by Corollary \ref{cor-trajectory-controllability}, we could regard the functional variables as several desired root nodes. More precisely, consider \texttt{BCN} $B(\tilde{{\bf G}},\tilde{{\bf F}},{\bf U})$ with outputs
\begin{equation}
{\bf y}_j(t)={\bm h}_j([{\bm x}_i(t)]_{i\in {\bf Y}_j}), ~j\in[n-p+1,n],
\end{equation}
where ${\bf Y}_j$ is the set of functional variables for logical function ${\bm h}_j$. If we regard nodes $v_j$ with $j\in\bigcup_{j=1}^{p}{\bf Y}_j$ as the root nodes, then the outputs can track to any target sequence according to Corollary \ref{cor-trajectory-controllability}. To this end, one can modify the set $\ddot{{\bf E}}$ given in the above subsection as $\ddot{{\bf E}}:=\vec{{\bf E}}\backslash {\bf E}'$ with ${\bf E}'=\{(v_i,v_j)\mid i\in \bigcup_{k=1}^{p}{\bf Y}_k, j \not\in \bigcup_{k=1}^{p}{\bf Y}_k\}$ to make this \texttt{BCN} track towards the given trajectory.

\section{Pinning Stabilization in probability}\label{sec-stability in probability}
In this section, we apply the structural controllability criterion--Theorem \ref{theorem-structuralcontrollability}--to overcome the difficulties of control design in \cite{huangc2020ins205} and \cite{huangc2020tnnls}, due to the lack of inclusion property for reachable subsets. Consequently, the pinning nodes and state feedback controllers in \cite{huangc2020ins205} were both given in advance. Or alternatively, only can the testification procedure be available. Although the time-varying state feedback controller was designed in \cite{huangc2020tnnls}, it can only guarantee the fixed-time reachability to the desirable stable state, and the control inputs after the reachability time is still lacking. As mentioned in \cite{huangc2020ins205,huangc2020tnnls}, the efficient algorithm for controller design is rare and the worst time complexity in the existing works on controller design is $\Theta((2^m)^{2^n})$, which is a severe computational burden for large-scale \texttt{BNs}. Besides, the obtained Theorem 3 in \cite{huangc2020tnnls} is still hard to use for designing time-varying controller after the system state reaching $x^\ast$ at the fixed time.

In this section, we will present an equivalence verification between several types of stability in probability. On this basis, the pinning control for stabilization in probability is designed by the structural controllability condition. Notice that the $2^n\times 2^n$-dimensional transition probability matrix (\texttt{TPM}), which is utilized in this section, is only applied to verify the equivalence between different types of stability. However, in the procedure of controller design, we still do not need to use the whole $2^n\times 2^n$-dimensional network transition matrix $L$ in (\ref{equ-assr-bn}), with the help of the structural controllability criterion.

By using \texttt{STP} of matrices, the \texttt{ASSR} of \texttt{PBN} (\ref{equ-pbn}) is developed as
\begin{equation}\label{equ-pbn-assr}
x(t+1)=\breve{L}^{\sigma(t)} x(t),
\end{equation}
where the evolution of state $x(t)$ can be equivalently regarded as the Markov chain $\{x(t)\in\Delta_{2^n}\mid t\in\mathbb{N}\}$ with \texttt{TPM} $\mathscr{P}=\sum_{i=1}^{s} p_i \breve{L}^{i}$, denoted by $M(\Delta_{2^n},\mathbb{N},\mathscr{P})$. Let $\breve{{\bf V}}=\Delta_{2^n}$, we can equivalently characterize the Markov chain $M(\Delta_{2^n},\mathbb{N},\mathscr{P})$ by the called state transition graph (\texttt{STG}) $\breve{{\bf G}}=(\breve{{\bf V}},\breve{{\bf E}},\breve{{\bf W}})$, where $(\delta_{2^n}^i,\delta_{2^n}^j) \in \breve{{\bf E}}$ if and only if $[\mathscr{P}]_{ji}>0$, and the weight function $\breve{{\bf W}}:\breve{{\bf E}}\rightarrow \mathbb{R}$ is defined as $\breve{{\bf W}}(\delta_{2^n}^i,\delta_{2^n}^j)=[\mathscr{P}]_{ji}$.

For the weighted digraph $\breve{{\bf G}}$, state $\delta_{2^n}^j$ is said to be reachable from state $\delta_{2^n}^i$ by one time step, denoted by $\delta_{2^n}^i\rightharpoonup\delta_{2^n}^j$, if $(\delta_{2^n}^i,\delta_{2^n}^j) \in \breve{{\bf E}}$. Moreover, state $\delta_{2^n}^j$ is said to be reachable from state $\delta_{2^n}^i$, denoted by $\delta_{2^n}^i\rightarrow\delta_{2^n}^j$, if there is a path from state $\delta_{2^n}^i$ to state $\delta_{2^n}^j$. On this basis, we denote $\delta_{2^n}^i\leftrightarrow\delta_{2^n}^j$ if $\delta_{2^n}^i\rightarrow\delta_{2^n}^j$ and $\delta_{2^n}^j\rightarrow\delta_{2^n}^i$.

\subsection{Definitions of Stability in Probability}
In this subsection, we review several definitions of stability in probability for \texttt{PBNs} in the literature.
\begin{definition}[See \cite{huangc2020ins205}\cite{huangc2020tnnls}]\label{def-SP}
Given state $x^\ast \in \Delta_{2^n}$, \texttt{PBN} (\ref{equ-pbn-assr}) is said to be stable in probability (\texttt{SP}) at state $x^\ast$ if, there exists an integer $T>0$ such that
\begin{equation}
{\bf P}\{x(t;x_0)=x^\ast\mid x(0)=x_0\}>0
\end{equation}
holds for any $x_0\in\Delta_{2^n}$ and any $t>T$.
\end{definition}

Accordingly, one can divide Definition \ref{def-SP} into the following two cases. Or alternatively, a \texttt{PBN} satisfying Definition \ref{def-SP} satisfies either Definition \ref{def-SAPP} or Definition \ref{def-SNP}.
\begin{definition}\label{def-SAPP}
Given state $x^\ast \in \Delta_{2^n}$, \texttt{PBN} (\ref{equ-pbn-assr}) is said to be stable in absolutely positive probability (\texttt{SAPP}) $\varepsilon$ at state $x^\ast$ if, \texttt{PBN} satisfies Definition \ref{def-SP} and there exists a number $0<\varepsilon\leq 1$, such
that for any $T>0$ we can find an integer $M>T$ satisfying
\begin{equation}
{\bf P}\{x(M;x_0)=x^\ast\mid x(0)=x_0\}>\varepsilon.
\end{equation}
\end{definition}

\begin{definition}\label{def-SNP}
Given state $x^\ast \in \Delta_{2^n}$, \texttt{PBN} (\ref{equ-pbn-assr}) is said to be stable in negligible probability (\texttt{SNP}) at state $x^\ast$ if, it satisfies Definition \ref{def-SP} and  $\lim\limits_{t\rightarrow+\infty}{\bf P}\{x(t;x_0)=x^\ast\mid x(0)=x_0\}=0$ holds for any $x_0\in\Delta_{2^n}$.
\end{definition}

\begin{remark}
In fact, Definition \ref{def-SNP} describes the scene of \texttt{SP} with probability $\lim\limits_{t\rightarrow+\infty}{\bf P}\{x(t;x_0)=x^\ast\mid x(0)=x_0\}=0$, while Definition \ref{def-SAPP} characterizes the case except that in Definition \ref{def-SNP}.
\end{remark}

\begin{definition}\label{def-SSPP}
Given state $x^\ast \in \Delta_{2^n}$, \texttt{PBN} (\ref{equ-pbn-assr}) is said to be stable in steady positive probability (\texttt{SSPP}) at state $x^\ast$ if, for any $x_0\in\Delta_{2^n}$, it holds that
\begin{equation}
\lim\limits_{t\rightarrow+\infty}{\bf P}\{x(t;x_0)=x^\ast\mid x(0)=x_0\}=\varepsilon_{x_0}
\end{equation}
exists for certain $0<\sigma_{x_0}\leq 1$.
\end{definition}
\begin{remark}
In this section, we will also show that once \texttt{PBN} (\ref{equ-pbn-assr}) satisfies Definition \ref{def-SSPP}, the number $\varepsilon_{x_0}$ is independent of initial state $x_0\in\Delta_{2^n}$.
\end{remark}

Finally, we provide the concept of general stability in probabilistic distribution (\texttt{SPD}) that was firstly proposed in \cite{zhusy2019tac}. Denote by $\nabla_{2^n}$ the set of all $2^n$-dimensional probability column vectors.

In this study, we simply modify the definition of stability in probability distribution (\texttt{SPD}) in \cite{zhusy2019tac} as follows.
\begin{definition}\label{def-SPD2}
\texttt{PBN} (\ref{equ-pbn-assr}) is said to be globally stable in probability distribution (\texttt{SPD}) w.r.t. ${\bm \mu}$ if, there exists a probability vector ${\bm \mu} \in \nabla_{2^n}$ satisfying that
\begin{equation}
\lim_{t\rightarrow +\infty}{\bf P}\left\{x(t;x_0)=\delta_{2^n}^j\mid x(0)=x_0\right\}= [{\bm \mu}]_j, ~j\in[1,2^n],
\end{equation}
for any $x_0\in\Delta_{2^n}$.
\end{definition}

\subsection{Some Lemmas of Markov Chains}
Here, some necessary conclusions of Markov chains are briefly introduced. In order to keep consistent with the expression of \texttt{PBNs}, the \texttt{TPM} that we consider here is a column-stochastic one, and the state space is assumed to be $\Delta_{2^n}$.

Given Markov chain $M(\Delta_{2^n},\mathbb{N},\mathscr{P})$. The first arrival probability from state $\delta_{2^n}^{i}$ to state $\delta_{2^n}^j$ at the $k$-th time step is denoted by $P^{(k)}_{i\rightarrow j}={\bf P}\{x(k)=\delta^j_{2^n}, x(t)\neq \delta_{2^n}^j, t\in[1,k-1] \mid x(0)=\delta_{2^n}^i\}$. Then, the first arrival probability from $\delta_{2^n}^{i}$ to $\delta_{2^n}^j$ is defined as $P_{i\rightarrow j}=\sum_{k=1}^{+\infty} P^{(k)}_{i\rightarrow j}$. To distinguish with the first arrival probability $P^{(k)}_{i \rightarrow j}$, we denote $\mathscr{P}^t$ by $\mathscr{P}(t)$.

As for the period of system states. Define the period of state $\delta_{2^n}^i$, represented by $d(\delta_{2^n}^i)$, as the largest common divisor of all integers $k$ that satisfy $[\mathscr{P}(k)]_{ii}>0$. State $\delta_{2^n}^i$ is said to be periodic if $d(\delta_{2^n}^i)>1$; otherwise (i.e., $d(\delta_{2^n}^i)=1$), it is called aperiodic. For the recurrence of states, state $\delta_{2^n}^i$ is called a recurrent (respectively, transition) state if $P_{i \rightarrow i}=1$ (respectively, $P_{i\rightarrow i}<1$).

\begin{lemma}[see \cite{markovchains}]\label{lem-markov}
\begin{itemize}
  \item[1)] If $\delta_{2^n}^\mu \leftrightarrow \delta_{2^n}^\nu$, then the types of $\delta_{2^n}^\mu$ and $\delta_{2^n}^\nu$ are the same, including periods and recurrence.

  \item[2)] Every Markov chain with finite state space contains at least one recurrent state.

  \item[3)] If state $\delta_{2^n}^\mu$ is recurrent and $\delta_{2^n}^\mu \rightarrow \delta_{2^n}^\nu$, then $\delta_{2^n}^\nu$ is recurrent and $P_{\nu\rightarrow\mu}=1$.

  \item[4)] If a Markov chain is an ergodic (i.e., recurrent and aperiodic) one, then it has the limiting distribution, which is exactly its unique steady distribution ${\bm \mu}$ satisfying ${\bm \mu}=\mathscr{P}{\bm \mu}$.

  \item[5)] If state $\delta_{2^n}^\mu$ is not recurrent, then $\lim\limits_{k\rightarrow+\infty} [\mathscr{P}(k)]_{\mu\nu}=0$ for any $\delta_{2^n}^\nu\in\Delta_{2^n}$.
\end{itemize}
\end{lemma}

\subsection{The Equivalence Among Stability in Probability}
Now, we establish the proof for the equivalence of Definition \ref{def-SP}, Definition \ref{def-SAPP}, Definition \ref{def-SSPP} and Definition \ref{def-SPD2}, which is crucial for the subsequent controller design.
\begin{theorem}\label{thm-equivalence}
Given state $x^\ast\in \Delta_{2^n}$, denoted by $x^\ast=\delta_{2^n}^{\alpha}$ without loss of generality, the following statements are equivalent:
\begin{itemize}
  \item[1)] \texttt{PBN} (\ref{equ-pbn-assr}) is \texttt{SP} at state $x^\ast$.
  \item[2)] \texttt{PBN} (\ref{equ-pbn-assr}) is \texttt{SAPP} at state $x^\ast$.
  \item[3)] \texttt{PBN} (\ref{equ-pbn-assr}) is \texttt{SSPP} at state $x^\ast$.
  \item[4)] \texttt{PBN} (\ref{equ-pbn-assr}) is \texttt{SPD} w.r.t. ${\bm \mu}\in \nabla_{2^n}$ with $[{\bm \mu}]_{\alpha}>0$.
\end{itemize}
\end{theorem}
\begin{proof}
Please refer to Appendix for the proof of this theorem.
\end{proof}

\begin{theorem}\label{thm-corollary}
Given $x^\ast\in\Delta_{2^n}$, \texttt{PBN} (\ref{equ-pbn-assr}) is \texttt{SP} at $x^\ast$ if and only if its \texttt{STG} $G$ contains an in-tree with the root $x^\ast$ and $x^\ast$ is an aperiodic state.
\end{theorem}
\begin{proof}
[Necessity] The proof of necessity is obvious. By Definition \ref{def-SP}, \texttt{STG} $G$ must contain an in-tree rooted at $x^\ast$. Besides, if we assume that $x^\ast$ is periodic, then one can conclude that
$P\{x(kd(x^\ast)-1)=x^\ast\mid x_0=x^\ast\}=0$. It contradicts with Definition \ref{def-SP}.

[Sufficiency] By the procedure in the proof of Theorem \ref{thm-equivalence}, if \texttt{STG} $G$ contains an in-tree rooted at $x^\ast$, then we can find the largest recurrent closed set $S$. Since $x^\ast$ is aperiodic,
set $S$ is an aperiodic and recurrent set. Thus, this \texttt{PBN} is \texttt{SP} at $x^\ast$. By Theorem \ref{thm-equivalence}, one can conclude that \texttt{PBN} (\ref{equ-pbn-assr}) is \texttt{SDP} w.r.t. ${\bm \mu}$ with $[{\bm \mu}]_{\alpha}>0$.
\end{proof}

\subsection{Stabilization in Probability}
To begin with, we first present a sufficient condition for the stability of \texttt{BNs} from the viewpoint of network structure, but the equilibrium point is not determined.
\begin{lemma}[See \cite{robert2012discrete}]\label{lemma-acylic}
Given \texttt{BN} $B({\bf G},{\bf F})$, it has a unique steady state if, its network structure ${\bf G}$ is acyclic.
\end{lemma}

Subsequently, given a state $x^\ast\in \Delta_{2^n}$, in order to realize the \texttt{SP} of \texttt{PBN} (\ref{equ-pbn}), we design the distributed pinning controller based on Theorem \ref{theorem-structuralcontrollability} by two steps.

\textbf{Step 1:} Globally Stabilizing A Mode. Assume that the network structure of every mode of \texttt{PBN} (\ref{equ-pbn-assr}) is cyclic. Otherwise, one directly goes to Step 2). We arbitrarily choose a mode $B(\breve{{\bf G}}^\iota,\breve{{\bf F}}^\iota)$, $\iota\in[1,s]$.

By resorting to Lemma \ref{lemma-acylic}, we would like to design the distributed pinning controller, which is imposed on the $\iota$-th mode $B(\breve{{\bf G}}^\iota,\breve{{\bf F}}^\iota)$, to make it stable. Observing Lemma \ref{lemma-acylic}, we can pick the pinning nodes for this mode into $\breve{\Gamma}^{\iota}$, in a similar manner as $\Gamma_1$. The pinning controller injected on the $\iota$-th mode can be designed as
\begin{equation}\label{equ-pinning-pbn-type1}
\left\{\begin{aligned}
{\bm x}_j(t+1)&={\bm u}_{\iota,j}(t) \oplus_{\iota,j} {\bm f}^\iota_j([{\bm x}_i(t)]_{i\in {\bf X}^\iota_j}),~j\in \breve{\Gamma}^\iota,\\
{\bm x}_j(t+1)&={\bm f}^\iota_j([{\bm x}_i(t)]_{i\in {\bf X}^\iota_j}),~j\not\in \breve{\Gamma}^\iota,
\end{aligned}\right.
\end{equation}
where ${\bm u}_{\iota,j}(t)={\bm g}^{\iota}_j([{\bm x}_i(t)]_{i\in {\bf X}^\iota_j})$ is the mode-based state feedback controller. Then, in order to design the logical operators $\oplus_{\iota,j}$ and feedback functions ${\bm g}^{\iota}_j$, we can solve the following equation to compute their corresponding structure matrices in a similar manner as those in Section \ref{sec-structuralcontrollability}:
\begin{equation}\label{equ-pbn-logicalequ}
L_{\oplus_{\iota,j}}L_{{\bm g}^{\iota}_j}(I_{2^{\mid \mathbf{X}^\iota_{j}\mid}}\otimes L_{{\bm f}^\iota_j})\Phi_{2^{\mid \mathbf{X}^{\iota}_{j}\mid }}=A^\iota_j (I_{2^{\mid \vec{{\bf N}}^{\text{in}}_j \mid}}\otimes {\bf 1}_{2^{\mid \mathbf{N}^{\text{in}}_k\backslash \vec{{\bf N}}^{\text{in}}_j \mid}}^\top) W^\top_{j,{\bf N}^{\text{in}}_j,\vec{{\bf N}}^{\text{in}}_j},j\in\breve{\Gamma}^{\iota},
\end{equation}
with logical matrices $A^\iota_k\in\mathscr{L}_{2\times 2^{\mid \vec{{\bf N}}^{\text{in}}_j \mid}}$. Since the network structure of the $\iota$-th mode, that is, \texttt{BN} (\ref{equ-pinning-pbn-type1}), is acyclic, by Lemma \ref{lemma-acylic}, one has a unique steady state $x^\iota_e$ by iterating $l^\iota$ times, where $l^\iota$ is equal to the length of the longest path in $\breve{{\bf G}}^\iota$.

\textbf{Step 2:} Forcing the Reachability From $x^\iota_e$ to $x^\ast$. In this step, we choose another mode $\vec{\iota}\neq\iota$. For mode $B(\breve{{\bf G}}^{\vec{\iota}},\breve{{\bf F}}^{\vec{\iota}})$, we can use the design approach of pinning controller, given in Section \ref{sec-pin-controllability}, to force the controllability. Thus, the last thing is to determine the inputs of its generators, the logical function ${\bm u}^{\vec{\iota}}_j(t)=\ddot{{\bm g}}^{\vec{\iota}}_j([{\bm x}_i(t)]_{i\in {\bf X}^{\vec{\iota}}_i})$.

Since the mode $B(\breve{{\bf G}}^{\vec{\iota}},\breve{{\bf F}}^{\vec{\iota}})$ is also structurally $\eta^{\vec{\iota}}$ fixed-time controllable, we can achieve the reachability from state ${\bm x}^\iota_e$ to state ${\bm x}^\ast$ within $\eta^{\vec{\iota}}$ time steps. Without loss of generality, we denote this process as ${\bm x}_{e}^{\iota} \stackrel{{\bm u}_1}{\rightharpoonup} {\bm x}_1 \stackrel{{\bm u}_2}{\rightharpoonup}\cdots \stackrel{{\bm u}_{\eta^{\vec{\iota}}}}{\rightharpoonup}{\bm x}^\ast$, and thus any feedback function $\ddot{{\bm g}}^{\vec{\iota}}_j$, $j\in[1,n]$, provided that this sequence is feasible.

\begin{theorem}
\texttt{BN} $B({\bf G},{\bf F})$ after \textbf{Step 1} and \textbf{Step 2} will be \texttt{SP} at $x^\ast$.
\end{theorem}
\begin{proof}
First of all, since there exists one mode $B(\breve{{\bf G}}^\iota,\breve{{\bf F}}^\iota)$ that has at least one unique steady state $x_e^\iota$, we have that $\delta_{2^n}^i \rightarrow x_e^\iota$, for any $\delta_{2^n}^i\in\Delta_{2^n}$. As for another mode $B(\breve{{\bf G}}^{\vec{\iota}},\breve{{\bf F}}^{\vec{\iota}})$, state $x^\ast$ is reachable from $x^\iota_{e}$. Thus, its \texttt{STG} contains an in-tree rooted at $x^\ast$ and it implies that $x^\ast$ must be recurrent. Again, because \texttt{STG} $\breve{{\bf G}}$ has a self loop $(x^\iota_e,x^\iota_e)$, state $x^\iota_e$ must be an aperiodic state. Due to $x^\iota_e \leftrightarrow x^\ast$, one has that $x^\ast$ is also an aperiodic state. According to Theorem \ref{thm-corollary}, one can conclude that this \texttt{PBN} is \texttt{SP} at $x^\ast$.
\end{proof}

\begin{remark}
Based on Theorem \ref{thm-equivalence}, the developed pinning control also suits the \texttt{SAPP}, \texttt{SSPP} as well as \texttt{SPD} w.r.t. ${\bm \mu}\in\nabla_{2^n}$. Therefore, it overcomes the drawback of control design in \cite{zhusy2019tac,huangc2020ins205,huangc2020tnnls}.
\end{remark}

\subsection{Discussions and Comparisons}
In Step 1), the process for searching all cycles can be implemented within time $\Theta(n^2)$. Again, the time to design the state feedback controller is bounded by $\Theta(n2^{3d^\ast_\iota})$. As for the determination process of equilibrium point $x^\iota_e$, we can search it with time $\Theta(l^\iota n 2^{d^\ast_\iota})$. Therefore, this step can be done within time $\Theta(n^2+n2^{3d^\ast_\iota}+l^\iota n 2^{d^\ast_\iota})$.

In Step 2), by the complexity analysis in Section \ref{sec-pin-controllability}, the complexity of controllability part is bounded by $\Theta(n2^{3d_{\vec{\iota}}^\ast}+2(n+m)^2)$, while the time to determine the inputs of generators is $\Theta(\eta^{\vec{\iota}}2^m2^{d^\ast_{\vec{\iota}}}+\eta^{\vec{\iota}}m2^{3d^\ast_{\vec{\iota}}})$. Thus, this step is limited by time $\Theta(n2^{3d_{\vec{\iota}}^\ast}+2(n+m)^2+\eta^{\vec{\iota}}2^m2^{d^\ast_{\vec{\iota}}}+\eta^{\vec{\iota}}m2^{3d^\ast_{\vec{\iota}}})$.

Besides, the pinning control nodes are designable in this paper rather than pre-assigned as in \cite{huangc2020ins205} and \cite{huangc2020tnnls}. Moreover, the designed pinning controller is still
effective even through it has reached the steady state $x^\ast$ with probability; but it cannot be well solved in the previous results in \cite{huangc2020ins205} and \cite{huangc2020tnnls}.

As a supplement, Theorem \ref{thm-equivalence} shows the equivalence between \texttt{SP} in \cite{huangc2020ins205}\cite{huangc2020tnnls} and \texttt{SDP} in \cite{zhusy2019tac}. We also show that if a \texttt{PBN} is \texttt{SP}, then the corresponding probability must approach to a positive constant, which is independent of initial states, with time tending to infinity.

\section{Simulations}\label{sec-simulations}
In this section, we would like to apply our theoretical results with the T-cell receptor kinetics with $37$ state nodes and $3$ control inputs \cite{klamt2006methodology} to illustrate the effectiveness. Its node dynamics can be described by equation (\ref{equation-exa-BN}), where three input nodes are $CD4$, $CD45$ and $TCRlig$, and the rest of nodes are all state nodes. Please refer to \cite{klamt2006methodology} for the detailed meaning of every abbreviation. Note that the out-degree of input ${\bm u}_1$ is more than one. Then, we add a virtual state variable as ${\bm x}_{38}(t+1)={\bm u}_1(t)$, and replace the node dynamics of ${\bm x}_{10}$ and ${\bm x}_{20}$ as ${\bm x}_{10}(\sharp)=({\bm x}_{20}(\ast) \wedge {\bm x}_{38}(\ast)) \vee ( {\bm x}_{35}(\ast) \wedge {\bm x}_{38}(\ast))$ and ${\bm x}_{20}(\sharp)=\overline{{\bm x}_{26}(\ast)}\wedge {\bm x}_{38}(\ast) \wedge {\bm u}_2(\ast)$. Correspondingly, one can draw its network structure ${\bf G}$ as Figure \ref{fig-exa-ns}.
\begin{figure*}[!ht]
\centering
\begin{equation}\label{equation-exa-BN}
{\scriptsize\begin{array}{lll}
\text{CD8}: {\bm u}_1, ~\text{CD45}: {\bm u}_2, ~\text{TCRlig}: {\bm u}_3, & \text{IKKbeta}: {\bm x}_{13}(\sharp)={\bm x}_{24}(\ast), & \text{PAGCsk}: {\bm x}_{26}(\sharp)={\bm x}_{10}(\ast) \vee \overline{{\bm x}_{35}(\ast)},\\
\text{AP1}: {\bm x}_1(\sharp)={\bm x}_9(\ast) \wedge {\bm x}_{18}(\ast), & \text{IP3}: {\bm x}_{14}(\sharp)={\bm x}_{25}(\ast), & \text{PLCg(bind)}: {\bm x}_{27}(\sharp)={\bm x}_{19}(\ast), \\
\text{Ca}: {\bm x}_2(\sharp)={\bm x}_{14}(\ast), & \text{Itk}: {\bm x}_{15}(\sharp)={\bm x}_{34}(\ast)\wedge {\bm x}_{37}(\ast), & \text{Raf}: {\bm x}_{28}(\sharp)={\bm x}_{29}(\ast), \\
\text{Calcin}: {\bm x}_3(\sharp)={\bm x}_2(\ast), & \text{IkB}: {\bm x}_{16}(\sharp)= \overline{{\bm x}_{13}(\ast)}, & \text{Ras}: {\bm x}_{29}(\sharp)={\bm x}_{12}(\ast) \vee {\bm x}_{30}(\ast), \\
\text{cCbl}: {\bm x}_4(\sharp)={\bm x}_{37}(\ast), & \text{JNK}: {\bm x}_{17}(\sharp)={\bm x}_{33}(\ast), & \text{RasGRP1}: {\bm x}_{30}(\sharp)={\bm x}_7(\ast) \wedge {\bm x}_{24}(\ast), \\
\text{CRE}: {\bm x}_5(\sharp)={\bm x}_6(\ast), & \text{Jun}: {\bm x}_{18}(\sharp)={\bm x}_{17}(\ast), & \text{Rlk}: {\bm x}_{31}(\sharp)={\bm x}_{20}(\ast), \\
\text{CREB}: {\bm x}_6(\sharp)={\bm x}_{32}(\ast), & \text{LAT}: {\bm x}_{19}(\sharp)={\bm x}_{37}(\ast), & \text{Rsk}: {\bm x}_{32}(\sharp) = {\bm x}_8(\ast), \\
\text{DAG}: {\bm x}_7(\sharp)={\bm x}_{25}(\ast), & \text{Lck}: {\bm x}_{20}(\sharp)=\overline{{\bm x}_{26}(\ast)}\wedge {\bm u}_1(\ast) \wedge {\bm u}_2(\ast), & \text{SEK}: {\bm x}_{33}(\sharp)={\bm x}_{24}(\ast),\\
\text{ERK}: {\bm x}_8(\sharp)={\bm x}_{21}(\ast), & \text{MEK}: {\bm x}_{21}(\sharp)={\bm x}_{28}(\ast),\text{NFAT}: {\bm x}_{22}(\sharp)={\bm x}_3(\ast), & \text{SLP76}: {\bm x}_{34}(\sharp)={\bm x}_{11}(\ast), \\
\text{Fos}: {\bm x}_9(\sharp)={\bm x}_8(\ast), & \text{NFkB}: {\bm x}_{23}(\sharp)=\overline{{\bm x}_{16}(\ast)}, ~\text{PKCth}: {\bm x}_{24}(\sharp)={\bm x}_7(\ast), & \text{TCRbind}: {\bm x}_{35}(\sharp)=\overline{{\bm x}_4(\ast)}\wedge {\bm x}_{35}(\ast),\\
\text{Fyn}: {\bm x}_{10}(\sharp) = ({\bm x}_{20}(\ast) \wedge {\bm u}_2(\ast)) \vee ({\bm x}_{35}(\ast) \wedge {\bm u}_2(\ast)), & \text{PLCg(act)}: {\bm x}_{25}(\sharp)=({\bm x}_{15}(\ast)\wedge {\bm x}_{27}(\ast) \wedge {\bm x}_{34}(\ast) & \text{TCRphos}: {\bm x}_{36}(\sharp)={\bm x}_{10}(\ast) \vee ({\bm x}_{20}(\ast) \wedge {\bm x}_{35}(\ast)), \\
\text{Gads}: {\bm x}_{11}(\sharp)={\bm x}_{19}(\ast),\text{Grb2Sos}: {\bm x}_{12}(\sharp)={\bm x}_{19}(\ast),& ~~~~\wedge {\bm x}_{37}(\ast))\vee ({\bm x}_{27}(\ast) \wedge {\bm x}_{31}(\ast) \wedge {\bm x}_{34}(\ast) \wedge {\bm x}_{37}(\ast)), & \text{ZAP70}: {\bm x}_{37}(\sharp)=\overline{{\bm x}_4(\ast)} \wedge {\bm x}_{20}(\ast) \wedge {\bm x}_{36}(\ast)
\end{array}}
\end{equation}
\hrulefill
\vspace*{2pt}
\end{figure*}

\begin{figure}[h!]
\centering
\begin{minipage}[c]{0.22\textwidth}
\subfigure[Network structure of \texttt{BCN} (\ref{equation-exa-BN}). ] {\label{fig-exa-ns} \includegraphics[height=2.2in]{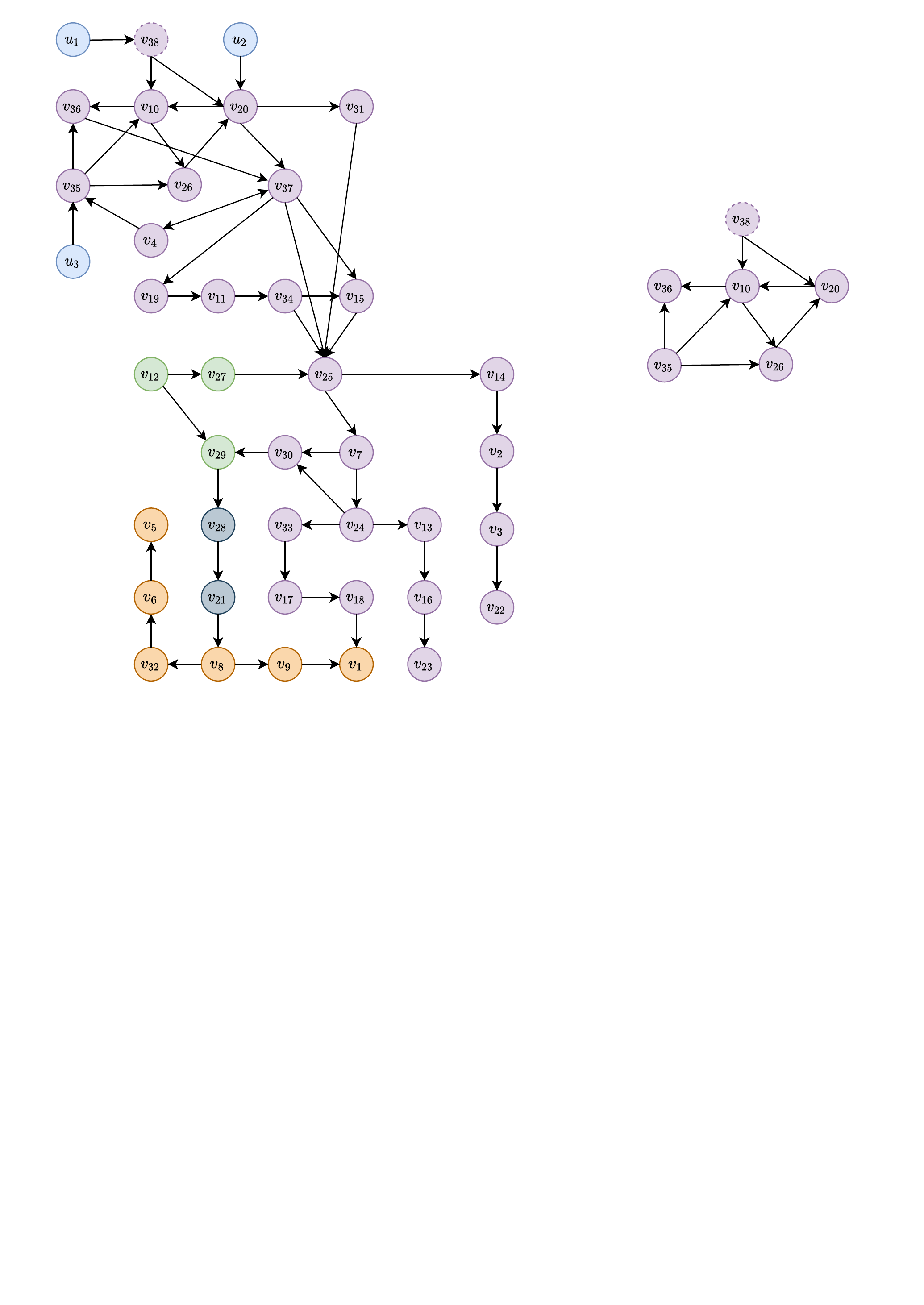} }
\end{minipage}
\hspace{0.02\textwidth}
\begin{minipage}[c]{0.22\textwidth}
\centering
\subfigure[The local network structure that need to be considered in the first aggregated subnetwork.] {\label{fig-exa-localbn} \includegraphics[height=1in]{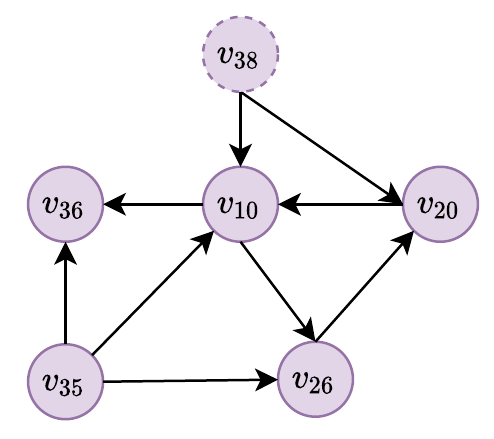} }
\end{minipage}
\end{figure}
\subsection{Minimum Control Node Problem}
Consider Problem \ref{problem-minimalnodecontrol} w.r.t. this \texttt{BCN}. We first construct the acyclic single-source-channel network aggregation in Definition \ref{def-aggregation}. Via the depth-first search algorithm, one can obtain all the strongly connected components, of which the unique non-trivial strongly connected component is only $\{v_4,v_{10},v_{20},v_{26},v_{35},v_{36},v_{37}\}$. Therefore, one can aggregate \texttt{BCN} (\ref{equation-exa-BN}) into four subnetworks as described in Figure \ref{fig-exa-ns}.

Consider subnetwork with the largest vertex set
$$\begin{aligned}{\bf N}_1=\{&u_1,u_2,u_3,x_{38},x_{36},x_{10},x_{20},x_{31},x_{35},x_{26},x_{37},x_4,x_{19},\\
&x_{11},x_{34},x_{15},x_{12},x_{27},x_{25},x_{14},x_{29},x_{30},x_{7},x_{2},x_{33},\\
&x_{24},x_{13},x_3,x_{22}\}.\end{aligned}$$
Then, we determine whether or not these are some vertices that have to be controlled, that is, set $W^1_1$ and $W^1_2$. According to 1) of Theorem \ref{thm-ctr-reduction}, we know that vertex $v_{12}$ must be controlled, i.e., $v_{12}\in W_1^1$. By 2) of Theorem \ref{thm-ctr-reduction}, one has that $\{v_2,v_3,v_{22},v_{11},v_{34}\}\subseteq W^1_2$ and $\{v_{35},v_{38}\}\subseteq W^1_2$. In terms of 3) of Theorem \ref{thm-ctr-reduction}, one can make $v_7\in W^1_1$, $v_{14}\in W^1_2$. With regard to 4) of Theorem \ref{thm-ctr-reduction}, one has that $v_{30}\in W^1_1$, $v_{24}\in W^1_2$, $v_{33}\in W^1_1$, $v_{13}\in W^1_2$, $v_{29}\in W^1_1$, $v_{27}\in W^1_2$, $v_{37}\in W_{1}^1$, $v_{31}\in W^1_2$, $v_{15}\in W^1_1$, $v_{25}\in W^1_1$, $v_{4}\in W^1_1$, $v_{19}\in W^1_2$. In the following, we only determine the control for the nodes $v_{36},v_{10},v_{20},v_{26}$. To this end, we only consider the logical network structure in Figure \ref{fig-exa-localbn}, and order its vertices as $\{v^1_1,v^1_2,v^1_3,v^1_4,v^1_5,v^1_6\}:=\{v_{10},v_{20},v_{26},v_{35},v_{36},v_{38}\}$. Consequently, we establish a $6 \times 6$-dimensional adjacency matrix as
\begin{equation}
A^1=(a^1_{ij})_{6 \times 6}=\left(\begin{array}{cccccc}
0&0&1&0&1&0\\
1&0&0&0&0&0\\
0&1&0&0&0&0\\
1&0&1&0&1&0\\
0&0&0&0&0&0\\
1&1&0&0&0&0
\end{array}
\right).
\end{equation}

To proceed, we define six binary variables $\varsigma^1_1,\varsigma^1_2,\varsigma^1_3,\varsigma^1_4,\varsigma^1_5,\varsigma^1_6$ to characterize the control for the above six vertices. Besides, since there is the unique cycle $\{v_{10},v_{26},v_{20}\}$, we have $\kappa=1$, $\gamma^1_{1,1}=\gamma^2_{1,1}=\gamma^3_{1,1}=1$ and $\gamma^4_{1,1}=\gamma^5_{1,1}=\gamma^6_{1,1}=0$ in (\ref{equ-optimal}). Given these parameters to solve (\ref{equ-optimal}), one has that
\begin{equation}
\begin{aligned}
J_1\ltimes\bar{M}^1=[&0,0,1,1,0,0,1,1,1,1,2,2,1,1,2,2,\\
&1,1,1,1,1,1,1,1,1,1,2,2,1,1,2,2,\\
&1,1,2,2,1,1,2,2,2,2,3,3,2,2,3,3,\\
&2,2,2,2,2,2,2,2,2,2,2,2,2,2,2,2]
\end{aligned}
\end{equation}
and
\begin{equation}
\begin{aligned}
J_1\ltimes\tilde{M}^1=[&1,1,1,1,1,1,1,1,1,1,1,1,1,1,1,1,\\
&1,1,1,1,1,1,1,1,1,1,1,1,1,1,1,1,\\
&1,1,1,1,1,1,1,1,1,1,1,1,1,1,1,1,\\
&1,1,1,1,1,1,1,1,0,0,0,0,0,0,0,0].
\end{aligned}
\end{equation}

Thus, by Theorem \ref{thm-solveoptimal}, we can obtain the feasible set $\zeta^1=\{\delta_{64}^i\mid i\in[1,10]\cup\{13,14,17,26,29,30,33,34,37,38\}\}$. We can check that $\delta_{64}^{30}$ corresponds to $[1,0,1,0,0,0]$ with the minimal number of $\Sigma_{i=1}^{6} \varsigma^1_i$. Totally, we can conclude that the minimum controlled node set for the first aggregation is $\Lambda_1^\ast=\{10,26,37,4,15,12,25,7,29,33\}$. By the same procedure, the minimum controlled node sets for the other two network aggregations are $\Lambda_2^\ast=\{32\}$ and $\Lambda_3^\ast=\emptyset$. Thus, the solution to Problem \ref{problem-minimalnodecontrol} for \texttt{BCN} (\ref{equation-exa-BN}) is $\Lambda^\ast=\{10,26,37,4,15,12,25,7,29,32,33\}$.

\subsection{Pinning Control Design for Controllability}
In this subsection, we would like to design a novel distributed pinning controller with lower time complexity so as to force the controllability of \texttt{BCN} (\ref{equation-exa-BN}).

Firstly, we pick the pinning node set $\Gamma$ via three steps. By the depth-first search algorithm, one can find $4$ cycles in total as $\{v_{37}, v_{4}\}$, $\{v_{37},v_4,v_{35},v_{26},v_{20}\}$, $\{v_{35},v_{10},v_{36},v_{37},v_4\}$ and $\{v_{10},v_{20},v_{26}\}$. Thus, in order to remove these cycles, we can select the first type of pinning nodes $\Gamma_1=\{v_{10},v_{37}\}$ and attempt to delete the directed edges $(v_4,v_{37})$, $(v_{20},v_{37})$, $(v_{36},v_{37})$ and $(v_{20},v_{10})$ to obtain the acyclic graph $\vec{{\bf G}}$.

As for the above acyclic graph $\vec{{\bf G}}$, the nodes with out-degree more than $1$ can be found as $v_{38},v_{35},v_{10},v_{37},v_{34},v_{7},v_8,v_{12},v_{24},v_{25}$. Correspondingly, their out-neighbor sets are respectively $\{v_{10},v_{20}\}$, $\{v_{10},v_{36}\}$, $\{v_{26},v_{36}\}$, $\{v_4,v_{19},v_{15},v_{25}\}$, $\{v_{15},v_{25}\}$, $\{v_{30},v_{37}\}$, $\{v_9,v_{32}\}$, $\{v_{27},v_{29}\}$, $\{v_{30},v_{33},v_{13}\}$, and $\{v_7,v_{14}\}$. Thus, one can define the function $\odot$ in a manner as (\ref{equ-odot}) as $\odot(v_{10})=\odot(v_{26})=\odot(v_{15})=\odot(v_{25})=\odot(v_{30})=2$, $\odot(v_{20})=\odot(v_{36})=\odot(v_4)=\odot(v_{19})=\odot(v_{37})=\odot(v_{9})=\odot(v_{32})=\odot(v_{27})=\odot(v_{29})=\odot(v_{33})=\odot(v_{13})=\odot(v_{7})=\odot(v_{14})=1$. Therefore, we can pick the second type of pinning nodes as $\Gamma_2=\{v_{10},v_{26},v_{15},v_{25},v_{30},v_{9},v_{27},v_{13},v_7\}$. Again, by deleting edges $(v_{38},v_{10})$, $(v_{35},v_{10})$, $(v_{35},v_{26})$, $(v_{10},v_{26})$, $(v_{37},v_{15})$, $(v_{37},v_{25})$, $(v_{37},v_4)$, $(v_{34},v_{15})$, $(v_7,v_{30})$, $(v_8,v_9)$, $(v_{12},v_{27})$, $(v_{24},v_{30})$, $(v_{24},v_{13})$ and $(v_{25},v_7)$. Then, the digraph $\ddot{{\bf G}}$ can be further obtained, where $\ddot{{\bf N}}^{\text{in}}_{10}=\ddot{{\bf N}}^{\text{in}}_{26}=\ddot{{\bf N}}^{\text{in}}_{15}=\ddot{{\bf N}}^{\text{in}}_{4}=\ddot{{\bf N}}^{\text{in}}_{30}=\ddot{{\bf N}}^{\text{in}}_{9}=\ddot{{\bf N}}^{\text{in}}_{27}=\ddot{{\bf N}}^{\text{in}}_{13}=\ddot{{\bf N}}^{\text{in}}_{7}=\ddot{{\bf N}}^{\text{in}}_{37}=\emptyset$, and $\ddot{{\bf N}}^{\text{in}}_{25}=\{v_{15},v_{27},v_{31},v_{34}\}$.

Subsequently, we collect the state nodes with out-degree zero in the digraph $\ddot{{\bf G}}$ as set $\Gamma_3=\{v_{10},v_{26},v_{15},v_{4},v_{30},v_{9},v_{27},v_{13},v_{7},v_{37}\}$. Consequently, one can obtain the pinning node set as $\Gamma=\Gamma_1\cup\Gamma_2\cup\Gamma_3=\{v_{10}, v_{15}, v_{25}, v_{37}, v_{30}, v_9, v_{27}, v_{13}, v_{7}, v_{26},v_{4}\}$, which is about $29.73\%$ of all state nodes.

Afterwards, we would like to design the state feedback control ${\bm g}_k$ and logical operator $\tilde{\oplus}_k$ for every node in $\Gamma_1\cup\Gamma_2$ to turn network structure ${\bf G}$ into $\ddot{{\bf G}}$. The design procedure for node $v_{25}$ is precisely introduced and those for other nodes in $\Gamma_1\cup\Gamma_2$ can be similarly obtained. By \texttt{STP} of matrices, one can establish the structure matrix of ${\bm f}_{25}$ as
$$\begin{aligned}L_{{\bm f}_{25}}=\delta_2[&1,2,2,2,1,2,1,2,2,2,2,2,2,2,2,2,\\
&1,2,2,2,2,2,2,2,2,2,2,2,2,2,2,2].\end{aligned}$$
To compute the logical matrix $F_{25}$ in equation (\ref{equation-equa}), let parameter $A_{25}=\delta_2[2,1,2,2,2,2,2,2,2,2,2,2,2,2,2,2]$, then there holds that
$$\begin{aligned}F_{25}=&A_{25}(I_{2^4}\otimes{\bf 1}_2^\top)W^\top_{[25,{\bf N}^{\text{in}}_{25},\ddot{{\bf N}}^{\text{in}}_{25}]}\\
=&A_{25}(I_{2^4}\otimes{\bf 1}_2^\top) (W^3_{[2,2^4]}W_{[2,2^3]})^\top\\
=&\delta_2[2,2,2,2,1,2,1,2,2,2,2,2,2,2,2,2,\\
&~~~~~~2,2,2,2,2,2,2,2,2,2,2,2,2,2,2,2].
\end{aligned}$$
In order to get ${\bm g}_k$ and $\tilde{\oplus}_k$, we solve the following equation
$$L_{\tilde{\oplus}_{25}}L_{{\bm g}_{25}}\left(I_{2^5}\otimes L_{{\bm f}_{25}}\right)\Phi_{2^{5}}=F_{25}.$$
Denote $L_{\tilde{\oplus}_{25}}=\left( \begin{array}{cccc}\alpha_1 &\alpha_2 &\alpha_3 &\alpha_4\\ 1-\alpha_1 &1-\alpha_2 &1-\alpha_3 &1-\alpha_4\end{array}\right)$ and $L_{{\bm g}_{25}}=\left( \begin{array}{cccc}\beta_1 &\beta_2 &\cdots &\beta_{25}\\ 1-\beta_1 &1-\beta_2 &\cdots &1-\beta_{25}\end{array}\right)$ and plug them into the above equation. We can calculate that one feasible solution is $\alpha_1=1$, $\alpha_2=\alpha_3=\alpha_4=0$, and $\beta_5=\beta_7=1$, $\beta_i=0$, $i\in[1,4]\cup\{6\}\cup[8,32]$. Correspondingly, we get $\tilde{\oplus}_{25}=\wedge$ and ${\bm g}_{25}(t)={\bm x}_{15}(t)\wedge{\bm x}_{27}(t)\wedge\overline{{\bm x}_{31}(t)}\wedge{\bm x}_{34}(t)$. In a similar manner, the logical operators and feedback functions for the other nodes in $\Gamma_1\cup\Gamma_2$ can be designed as
\begin{equation*}
\left\{\begin{aligned}
&\tilde{\oplus}_{4}=\vee, ~{\bm g}_4(\ast)=\overline{{\bm x}_{37}(\ast)}, \\
&\tilde{\oplus}_{7}=\vee, ~{\bm g}_{7}(\ast)=\overline{{\bm x}_{25}(\ast)},\\
&\tilde{\oplus}_{9}=\vee, ~{\bm g}_{9}(\ast)={\bm x}_8(\ast),~\tilde{\oplus}_{30}=\vee, \\
&{\bm g}_{30}(\ast)=\overline{{\bm x}_7(\ast) \wedge {\bm x}_{24}(\ast)},\\
&\tilde{\oplus}_{10}=\vee, ~{\bm g}_{10}(\ast)=\overline{({\bm x}_{20}(\ast) \wedge {\bm x}_{38}(\ast)) \vee ({\bm x}_{35}(\ast) \wedge {\bm x}_{38}(\ast))},\\
&\tilde{\oplus}_{13}=\vee, ~{\bm g}_{13}(\ast)=\overline{{\bm x}_{24}(\ast)},\\
&\tilde{\oplus}_{15}=\vee, ~{\bm g}_{15}(\ast)=\overline{{\bm x}_{34}(\ast)\wedge {\bm x}_{37}(\ast)},\\
&\tilde{\oplus}_{26}=\vee, ~{\bm g}_{26}(\ast)=\overline{{\bm x}_{10}(\ast) \vee \overline{{\bm x}_{35}(\ast)}},\\
&\tilde{\oplus}_{27}=\vee, ~{\bm g}_{27}(\ast)=\overline{{\bm x}_{19}(\ast)},\\
&\tilde{\oplus}_{37}=\vee, ~{\bm g}_{37}(\ast)=\overline{\overline{{\bm x}_4(\ast)} \wedge {\bm x}_{20}(\ast) \wedge {\bm x}_{36}(\ast)}.
\end{aligned}
\right.
\end{equation*}

Finally, consider the pinning nodes in the set $\Gamma_3$. One can directly inject the open-loop inputs $\tilde{{\bm u}}_k$ by logical operator $\wedge$ as in (\ref{equation-pinning-bn}).

By contrast, if utilizing the traditional \texttt{ASSR} approach, then we need to handle a $2^{37}\times 2^{40}$-dimensional network transition matrix in (\ref{equ-assr-bcn}). Under our framework, the dimension of the considered matrix is only $2\times 2^5\ll 2^{37}\times 2^{40}$.

\subsection{Stabilization in Probability}
As is well known, gene mutation is a usual phenomenon in gene regulatory networks. We assume that the possible mutation positions are $v_4$ and $v_{37}$ with possibility $0.5$ and $0.4$, respectively. Once the mutation happens, their node dynamics respectively turns to ${\bm x}_4(\sharp)={\bm x}_{13}(\ast)$ and ${\bm x}_{37}(\sharp)=0$. Thus, there are four distinct modes corresponding to the cases that genes $v_4$ and $v_{37}$ do not mutate, both $v_4$ mutates and $v_{37}$ does not mutate, $v_4$ does not mutate and $v_{37}$ mutates, $v_4$ and $v_{37}$ mutate. Afterwards, we design the pinning controller to stabilize this \texttt{PBN} at ${\bm x}^\ast:=(\delta_{37}^{13}+\delta_{37}^{15}+\delta_{37}^{16})^\top$ in probability.

Step 1: Globally stabilizing the second mode with ${\bm x}_4(\sharp)={\bm x}_{13}(\ast)$. Then, the unique cycle in its network structure is $\{v_{10},v_{20},v_{26}\}$. Thus, for this mode, one can choose the pinning node set as $\Gamma^1=\{v_{26}\}$. Again, the logical operator $\oplus_{1,26}$ and ${\bm g}^1_{26}$ in (\ref{equ-pinning-pbn-type1}) aim to remove the functional variable ${\bm x}_{10}$ from ${\bm f}_{26}$. By solving equation (\ref{equ-pbn-logicalequ}) as that in the above subsection, one can obtain that $L_{\oplus_{1,26}}=\left( \begin{array}{cccc} 1 &0 &0 &0\\ 0 &1 &1 &1 \end{array}\right)$ and $L_{{\bm g}^1_{26}}=\left( \begin{array}{cccc} 1 &1 &0 &0\\ 0 &0 &1 &1 \end{array}\right)$, which respectively correspond to $\oplus_{1,26}=\wedge$ and ${\bm g}^1_{26}=\overline{{\bm x}_{35}(\ast)}$. Let ${\bm u}_1={\bm u}_2={\bm u}_3=1$ in this mode. Then, one can compute its \texttt{STG} and the unique attractor ${\bm x}^2_e=(\delta_{37}^{16}+\delta_{37}^{26})^\top$ by R as Fig. \ref{fig-exa-ns}.

\begin{figure}[h!]
\centering
\begin{minipage}[c]{0.22\textwidth}
\subfigure[The original \texttt{STG} of mode $2$. ] {\label{fig-exa-ns} \includegraphics[height=1.6in]{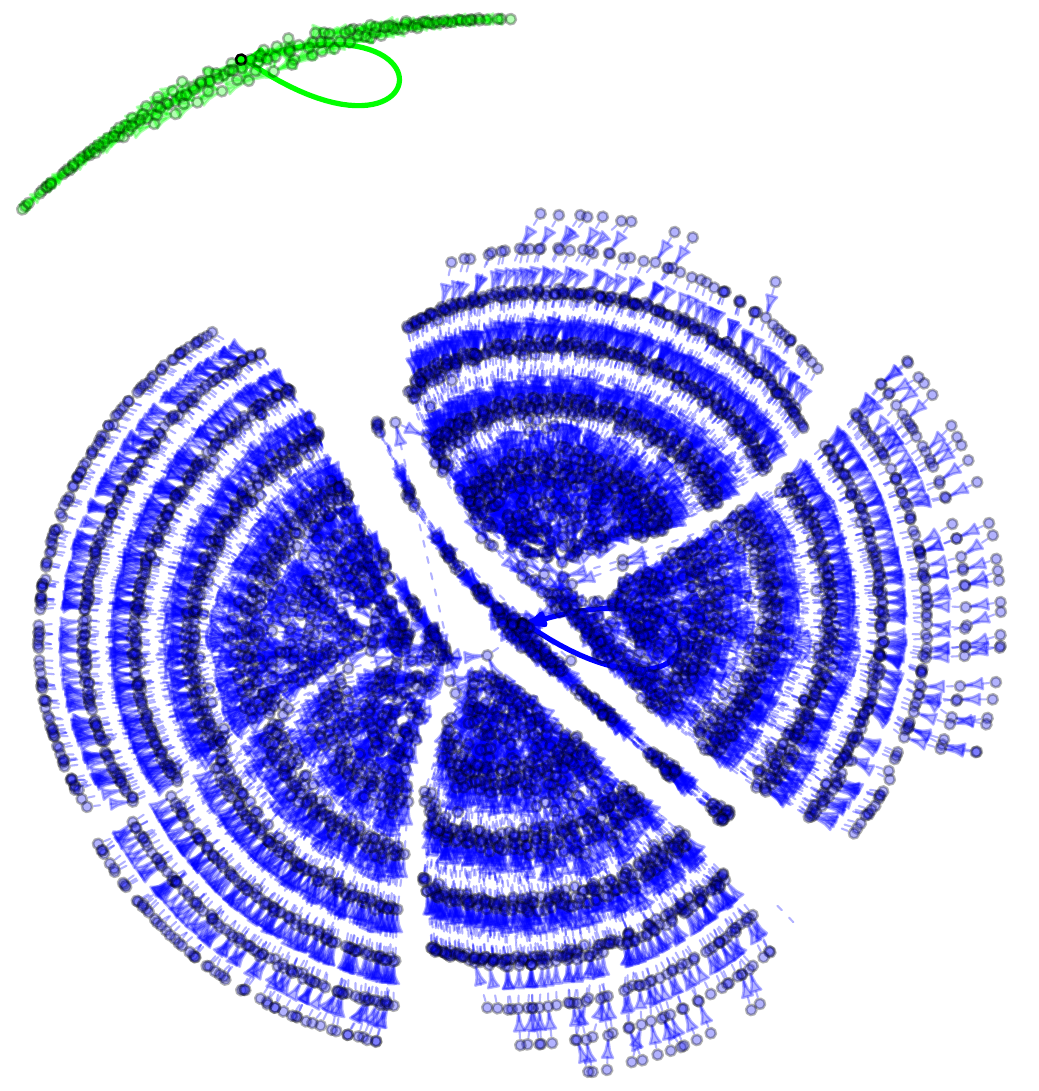} }
\end{minipage}
\hspace{0.01\textwidth}
\begin{minipage}[c]{0.22\textwidth}
\centering
\subfigure[The \texttt{STG} of mode $2$ after Step 1).] {\label{fig-exa-localbn} \includegraphics[height=1.6in]{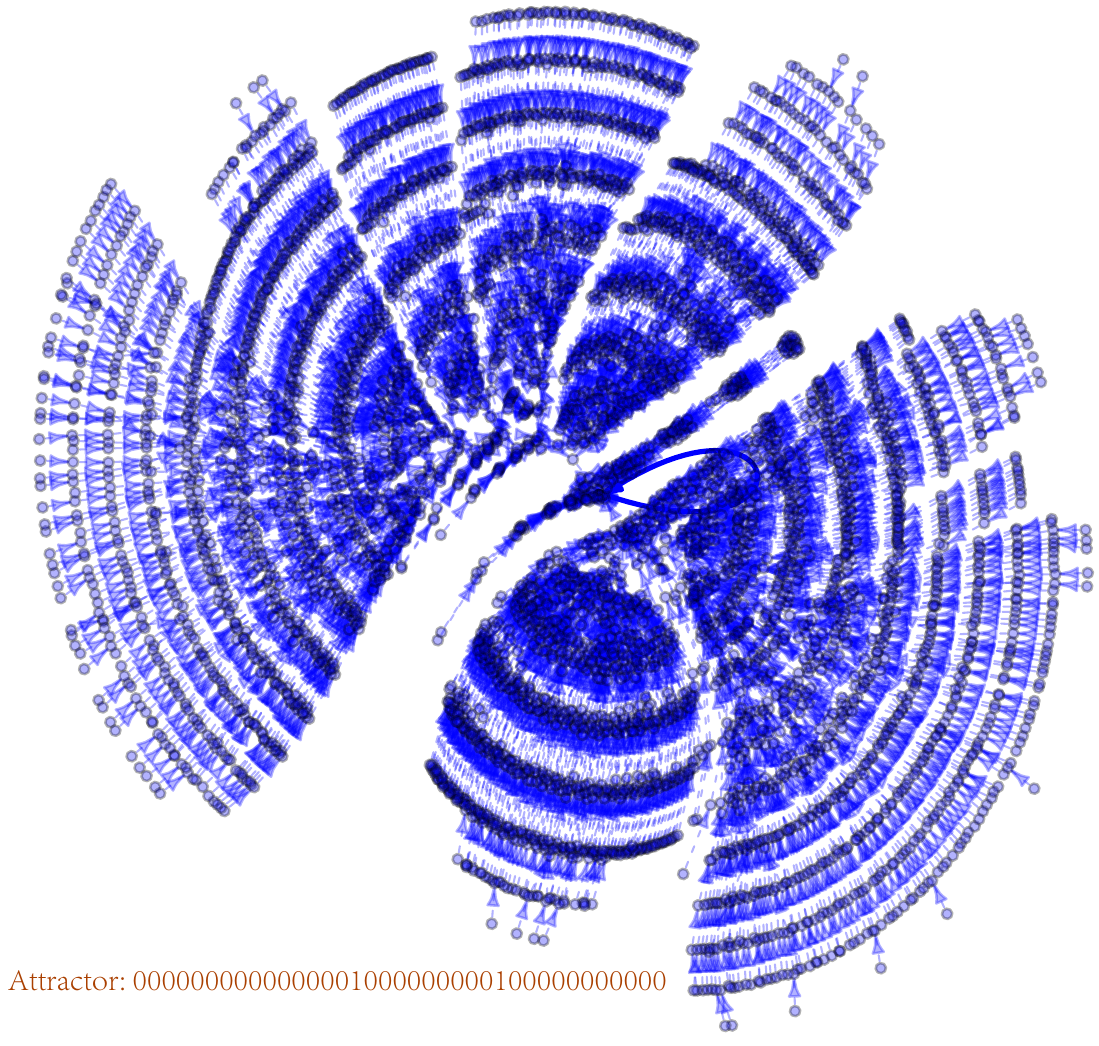} }
\end{minipage}
\end{figure}

Step 2: Forcing the reachability from ${\bm x}^2_e$ to ${\bm x}^\ast$. Consider the first mode, for which the pinning control designed in the above subsection for controllability is still available. Denote by $\tilde{{\bm u}}=({\bm u}_1,{\bm u}_2,{\bm u}_{3},\tilde{{\bm u}}_4,\tilde{{\bm u}}_7,\tilde{{\bm u}}_9,\tilde{{\bm u}}_{10},\tilde{{\bm u}}_{13},\tilde{{\bm u}}_{15},\tilde{{\bm u}}_{26},\tilde{{\bm u}}_{27},\tilde{{\bm u}}_{30},\tilde{{\bm u}}_{37})$. One can calculate the trajectory from ${\bm x}^2_e$ to ${\bm x}^\ast$ as ${\bm x}^2_e \stackrel{ (\Sigma_{i=1}^{4}\delta_{13}^{i})^\top }{\rightharpoonup} (\delta_{37}^4+\delta_{37}^{16})^\top \stackrel{ (\delta_{13}^{8}+\delta_{13}^{9})^\top }{\rightharpoonup} {\bm x}^\ast$. Thus, one of feasible feedback controllers for the input ${\bm u}$ can be given as
\begin{equation*}
\begin{aligned}
{\bm u}_1(\ast)&={\bm x}_{26}(\ast), ~{\bm u}_{2}(\ast)={\bm x}_{26}(\ast), ~{\bm u}_{3}(\ast)={\bm x}_{26}(\ast), \\
\tilde{{\bm u}}_{4}(\ast)&={\bm x}_{26}(\ast), ~\tilde{{\bm u}}_{13}(\ast)={\bm x}_4(\ast),~\tilde{{\bm u}}_{15}(\ast)={\bm x}_4(\ast),\\
\tilde{{\bm u}}_{7}(\ast)&=\tilde{{\bm u}}_{9}(\ast)=\tilde{{\bm u}}_{10}(\ast)=\tilde{{\bm u}}_{26}(\ast)\\
&=\tilde{{\bm u}}_{27}(\ast)=\tilde{{\bm u}}_{30}(\ast)=\tilde{{\bm u}}_{37}(\ast)=0.
\end{aligned}
\end{equation*}

\section{Conclusion}\label{sec-conclusion}
To overcome Problems ({\bf P1})-({\bf P6}) in the existing results, this paper has provided a novel and general control framework for \texttt{BNs}. Without utilizing the node dynamics, the structural controllability of \texttt{BCNs} has been formalized for the first time. Moreover, the minimum node control problem for \texttt{BNs} w.r.t. the structural controllability has been proved to be NP-hard. Furthermore, an efficient network aggregation has been proposed based on the structural controllability criterion; it also answered an open problem in \cite{zhaoy2015tnnls}. Based on this aggregation approach, all feasible solutions to the minimum node control problem has been found. This structural controllability condition has also been applied to the cases where the node dynamics is identifiable. For one thing, a feasible pinning strategy has been given to force the controllability of an arbitrary \texttt{BN}, where the time complexity is dramatically reduced to $\Theta(n2^{3d^\ast}+2(n+m)^2)$ and the pinning nodes can be easily selected. For the other thing, an interesting theorem has been presented to show the equivalence among several types of stability in probability for \texttt{PBNs}. By utilizing this condition, we have provided an efficient procedure to design the feasible control strategy.


\section*{Appendix\\Proof of Theorem \ref{thm-equivalence}}
\begin{proof}[Proof of Theorem \ref{thm-equivalence}]
The verification of 4) $\Rightarrow$ 3) $\Rightarrow$ 2) $\Rightarrow$ 1) can be followed by their definitions directly. Thus, if statement 1) can be proved to imply statement 4), then the proof of Theorem \ref{thm-equivalence} can be completed.

According to Definition \ref{def-SP}, for any $x_0\in \Delta_{2^n}$, we can find an integer $T_{x_0}\leq 2^n$ such that ${\bf P}\{x(t;x_0)=x^\ast\mid x(0)=x_0\}>0$, for any $t\geq T_{x_0}$. Then, we respectively discuss this problem by two cases.

Case I): State $x^\ast$ is a common fixed point of each mode, that is, ${[\mathscr{P}]}_{\alpha\alpha}=1$. For this case, it is equivalent to Theorem 1 in \cite{guo2018tacPBNstability}. Thus, one can imply that $\lim\limits_{t\rightarrow+\infty}{\bf P}\{x(t)=x^\ast\mid x(0)=x_0\}=1$, which indicates that Definition \ref{def-SSPP} is satisfied.

Case II): State $x^\ast$ is not a common fixed point of all modes, that is, $[\mathscr{P}]_{\alpha\alpha}<1$. In this case, there must exist another state $\delta_{2^n}^{g_1}\in\Delta_{2^n}$ satisfying $[\mathscr{P}]_{g_1\alpha}>0$. Besides, for state $\delta_{2^n}^{g_1}$, we can find an integer $\kappa_1$ with $(\delta_{2^n}^{\alpha})^\top \sum_{\tau=1}^{\kappa_1-1}[\mathscr{P}(\tau)] \delta_{2^n}^{g_1}=0$ and $(\delta_{2^n}^{\alpha})^\top \sum_{\tau=1}^{\kappa_1}[\mathscr{P}(\tau)] \delta_{2^n}^{g_1}>0$. Then, we can print the path $\delta_{2^n}^{g'_0}:=\delta_{2^n}^{g_1} \rightharpoonup \delta_{2^n}^{g'_2} \rightharpoonup \cdots \rightharpoonup \delta_{2^n}^{g'_{\kappa_1}}:=\delta_{2^n}^{\alpha}$. Let $S_1=\{ \delta_{2^n}^{g_0'},\delta_{2^n}^{g'_1},\cdots,\delta_{2^n}^{g'_{\kappa_1}}\}$. Secondly, if it holds that $J_{\Delta_{2^n} \backslash S_1} \times \mathscr{P} \times J_{S_1}>0$, then we can find states $\delta_{2^n}^{\bar{g}_2} \in S_1$ and $\delta_{2^n}^{g_2}\in\Delta_{2^n} \backslash S_1$ with $[\mathscr{P}]_{g_2\bar{g}_2}>0$. Since $\delta_{2^n}^{g_2}\rightarrow x^\ast$, there exists an integer $\kappa_2\in\mathbb{N}$ such that $J_{S_1} \ltimes \sum_{\tau=1}^{\kappa_2-1}[\mathscr{P}(\tau)] \delta_{2^n}^{g_2}=0$ and $J_{S_1} \ltimes \sum_{\tau=1}^{\kappa_2}[\mathscr{P}(\tau)] \delta_{2^n}^{g_2}>0$. Similarly, we can define a path $\delta_{2^n}^{\tilde{g}_0}:=\delta_{2^n}^{g_2} \rightharpoonup \delta_{2^n}^{\tilde{g}_1} \rightharpoonup \cdots \rightharpoonup \delta_{2^n}^{\tilde{g}_{\kappa_2}}\in S_1$, whose vertices are collected into set $S_2=\{\delta_{2^n}^{\tilde{g}_0},\delta_{2^n}^{\tilde{g}_1},\cdots,\delta_{2^n}^{\tilde{g}_{\kappa_2}} \}$. This procedure would be implemented until that set $S_k$ cannot reach the states in $\Delta_{2^n} \backslash S$, where $S=\bigcup_{i=1}^{k} S_i$.

For this set $S$, we prove that any state pair $\delta_{2^n}^\mu,\delta_{2^n}^\nu\in S$ is mutually reachable, that is, $\delta_{2^n}^\mu \leftrightarrow \delta_{2^n}^\nu$. For the construction of $S$, it holds that $\delta_{2^n}^\mu\rightarrow x^\ast$, $\delta_{2^n}^\nu\rightarrow x^\ast$, $x^\ast\rightarrow\delta_{2^n}^\mu$, and $x^\ast\rightarrow\delta_{2^n}^\nu$. Thus, one can conclude that $\delta_{2^n}^\mu \leftrightarrow \delta_{2^n}^\nu$, for any state pair $\delta_{2^n}^\mu,\delta_{2^n}^\nu\in S$. By conclusion 1) of Lemma \ref{lem-markov}, the period and recurrence of all states in set $S$ are the same with each other.

In the following, we prove that all the states in set $S$ are recurrent and aperiodic. By conclusion 2) of Lemma \ref{lem-markov}, there must exist a recurrent state, denoted by $\tilde{x}$. Since $\tilde{x}\rightarrow x^\ast$, by conclusion 3) of Lemma \ref{lem-markov}, one can imply that state $x^\ast$ is recurrent. Hence, all other states in set $S$ are also recurrent in set $S$. Subsequently, we show that the states in set $S$ are all aperiodic. Supposing that state $x^\ast$ is a periodic state, it means that the maximum common divisor of the integer set $\{k\mid [\mathscr{P}(k)]_{\alpha\alpha}>0\}$ satisfies $d(\alpha)>1$. Thereby, for any large number $M>0$, one can choose an arbitrary integer $\hbar_{\alpha}>\frac{M}{d(\alpha)}$ such that $[\mathscr{P}(\hbar_{\alpha}d(\alpha))]_{\alpha\alpha}=0$; it contradicts with Definition \ref{def-SP}. Therefore, one can conclude that all states in set $S$ are both recurrent and aperiodic.

Afterwards, we prove that the set $S$ contains all the recurrent and aperiodic states of Markov chain $M(\Delta_{2^n},\mathbb{N},\mathscr{P})$. If there is a recurrent state $\delta_{2^n}^\beta\in\Delta_{2^n} \backslash S$ with $\delta_{2^n}^\beta \rightarrow S$, then it implies that $P_{\gamma\rightarrow\beta}=1$ for any $\delta_{2^n}^\gamma\in S$ by conclusion 3) of Lemma \ref{lem-markov}. That is, we have $\delta_{2^n}^\gamma \rightarrow \delta_{2^n}^\beta$. It contradicts with the fact that the states in set $S$ cannot reach those in $\Delta_{2^n}\backslash S$ anymore.

Without loss of generality, we assume that set $S=\{\delta_{2^n}^1,\delta_{2^n}^2,\cdots,\delta_{2^n}^s\}$. Because set $S$ cannot reach any state in $\Delta_{2^n}\backslash S$, one can split the matrix $\mathscr{P}$ as $\mathscr{P}=\left( \begin{array}{cc} \mathscr{P}_1 & \mathscr{P}_2 \\ {\bf 0}_{(2^n-s) \times s} & \mathscr{P}_3 \end{array} \right)$. It means that matrix $\mathscr{P}_1$ is well defined as a column-stochastic one.

According to conclusion 4) of Lemma \ref{lem-markov}, Markov chain $M(\Delta_{2^n-s},\mathbb{N},\mathscr{P}_1)$ is ergodic, thus it has a unique steady distribution $\vec{\omega}$ satisfying $\vec{\omega}=\mathscr{P}_1\vec{\omega}$. Then, by extending the vector $\vec{\omega}$ to $\omega=[\vec{\omega}^\top,{\bf 0}_{2^n-s}^\top]^\top$, we prove that the Markov chain $M(\Delta_{2^n-s},\mathbb{N},\mathscr{P}_1)$ has the unique limiting distribution.

By conclusion 4) of Lemma \ref{lem-markov}, it holds that $\lim\limits_{t\rightarrow+\infty}\mathscr{P}_3(t)={\bf 0}_{(2^n-s)\times (2^n-s)}$. Let $\delta=\kappa+\eta$, it holds that
\begin{equation}
\begin{aligned}
\mathscr{P}(\delta)&=\mathscr{P}(\kappa)\mathscr{P}(\eta)\\
&=\left( \begin{array}{cc} \mathscr{P}_1(\kappa)\mathscr{P}_1(\eta) & \mathscr{P}_1(\kappa)\mathscr{P}_2(\eta)+\mathscr{P}_2(\kappa)\mathscr{P}_3(\eta) \\ {\bf 0}_{(2^n-s) \times s} & \mathscr{P}_3(\kappa+\eta) \end{array}\right)\\
&=\left( \begin{array}{cc} \mathscr{P}_1(\delta) & \mathscr{P}_2(\delta) \\ {\bf 0}_{(2^n-s) \times s} & \mathscr{P}_3(\delta) \end{array}\right).
\end{aligned}
\end{equation}

Since it has been proved that $\lim\limits_{\delta\rightarrow+\infty}\mathscr{P}_1(\delta)={\bf 1}_s^\top \otimes \vec{\omega}$ and $\lim\limits_{\delta\rightarrow+\infty}\mathscr{P}_3(\delta)={\bf 0}_{(2^n-2)\times(2^n-s)}$, we only need to verify that $\lim\limits_{\delta \rightarrow + \infty} \mathscr{P}_2(\delta)={\bf 1}^\top_{s}\otimes \vec{\omega} $, i.e., $\lim\limits_{\delta \rightarrow + \infty} \| \mathscr{P}_2(\delta)-{\bf 1}^\top_{s}\otimes \vec{\omega} \|=0$. Firstly, it holds that
\begin{equation*}
\begin{aligned}
&\| \mathscr{P}_2(\delta)-{\bf 1}^\top_{s}\otimes \vec{\omega} \|\\
&=\| \mathscr{P}_1(\kappa)\mathscr{P}_2(\eta)+\mathscr{P}_2(\kappa)\mathscr{P}_3(\eta)-{\bf 1}^\top_{s}\otimes \vec{\omega} \|\\
&=\|\mathscr{P}_1(\kappa)\mathscr{P}_2(\eta)-({\bf 1}^\top_{s}\otimes \vec{\omega}) \mathscr{P}_2(\eta) \\
&~~~~~~~~~~~~~ + ({\bf 1}^\top_{s}\otimes \vec{\omega}) \mathscr{P}_2(\eta) - {\bf 1}^\top_{s}\otimes \vec{\omega} + \mathscr{P}_2(\kappa)\mathscr{P}_3(\eta)\|\\
&\leq \|\mathscr{P}_1(\kappa)\mathscr{P}_2(\eta)-({\bf 1}^\top_{s}\otimes \vec{\omega}) \mathscr{P}_2(\eta)\| \\
&~~~~~~~~~~~~~ + \parallel({\bf 1}^\top_{s}\otimes \vec{\omega}) \mathscr{P}_2(\eta) - {\bf 1}^\top_{s}\otimes \vec{\omega} \parallel + \parallel\mathscr{P}_2(\kappa)\mathscr{P}_3(\eta)\|.
\end{aligned}
\end{equation*}

Subsequently, we define three items as
$$\chi_1(\delta)=\|\mathscr{P}_1(\kappa)\mathscr{P}_2(\eta)-({\bf 1}^\top_{s}\otimes \vec{\omega}) \mathscr{P}_2(\eta)\|,$$
$$\chi_2(\eta)=\parallel({\bf 1}^\top_{s}\otimes \vec{\omega}) \mathscr{P}_2(\eta) - {\bf 1}^\top_{s}\otimes \vec{\omega} \parallel,$$
and
$$\chi_3(\delta)=\parallel\mathscr{P}_2(\kappa)\mathscr{P}_3(\eta)\|.$$

Since $\lim\limits_{\delta\rightarrow+\infty}\mathscr{P}_1(\delta)={\bf 1}_s^\top \otimes \vec{\omega}$, for any $\hat{\varepsilon}=\frac{\varepsilon}{3}>0$, one can find an integer $N_1(\hat{\varepsilon})>0$ such that
$ \parallel\mathscr{P}_1(\kappa)-{\bf 1}_s^\top \otimes \omega \parallel \leq \hat{\varepsilon}$
holds for $\kappa \geq N_1(\hat{\varepsilon})$. With regard to $\chi_1(\delta)$, for such $\hat{\varepsilon}=\frac{\varepsilon}{3}>0$, let $\delta\geq N_1(\hat{\varepsilon})$, and then one has that
\begin{equation}
\begin{aligned}
&\|\mathscr{P}_1(\kappa)\mathscr{P}_2(\eta)-({\bf 1}^\top_{s}\otimes \vec{\omega}) \mathscr{P}_2(\eta)\|\\
& =\|[\mathscr{P}_1(\kappa)-({\bf 1}^\top_{s}\otimes \vec{\omega})]\mathscr{P}_2(\eta)\|\\
& \leq \|\mathscr{P}_1(\kappa)-({\bf 1}^\top_{s}\otimes \vec{\omega})\parallel \times \parallel\mathscr{P}_2(\eta)\| \leq \hat{\varepsilon}=\frac{\varepsilon}{3}.
\end{aligned}
\end{equation}

As for item $\chi_2(\eta)$, we rewrite the matrix $\mathscr{P}_2(t)$ as
\begin{equation}
\mathscr{P}_2(t)=\left(\mathscr{P}_2^{i,j}(t)\right)_{s \times (2^n-s)}.
\end{equation}
Because $\lim\limits_{\delta\rightarrow+\infty}\mathscr{P}_3(\delta)={\bf 0}_{(2^n-2)\times(2^n-s)}$ and $\parallel\text{Col}_j(\mathscr{P}(\delta))\parallel=1$, one can imply that for any $\bar{\varepsilon}=\frac{\varepsilon}{3(2^n-s)}>0$, there exists an integer $N_2(\bar{\varepsilon})>0$ satisfying that for any $\delta\geq N_2(\bar{\varepsilon})$ and any $j\in[1,2^n-s]$, it holds that
$\mid \parallel\text{Col}_j(\mathscr{P}_2(\delta))\parallel-1 \mid \leq \bar{\varepsilon}.$
It is equivalent to $\sum_{j=1}^{2^n-s}\mid\sum_{i=1}^{s}\mathscr{P}_2^{i,j}(\delta)-1\mid \leq (2^n-s)\bar{\varepsilon}=\frac{\varepsilon}{3}$. Therefore, for such $\bar{\varepsilon}$ and $\eta \geq N_2(\bar{\varepsilon})$, one has that
\begin{equation}
\begin{aligned}
\chi_2&=\parallel ({\bf 1}^\top_{s} \otimes \vec{\omega}) \mathscr{P}_2(\eta) - ({\bf 1}^\top_{s} \otimes \vec{\omega}) \parallel\\
&\leq \parallel ({\bf 1}^\top_{s} \otimes \vec{\omega}) \sum_{j=1}^{2^n-s}\sum_{i=1}^{s}(\mathscr{P}_2^{i,j}(\eta)-1) \parallel \leq (2^n-s)\bar{\varepsilon}=\frac{\varepsilon}{3}.
\end{aligned}
\end{equation}

For item $\chi_3(\delta)$. Since $\lim\limits_{\delta\rightarrow+\infty}\mathscr{P}_3(\delta)={\bf 0}_{(2^n-2)\times(2^n-s)}$, for any $\vec{\varepsilon}=\frac{\varepsilon}{3}>0$, there exists an integer $N_3(\vec{\varepsilon})>0$ such that $\|\mathscr{P}_3(\delta)\|<\frac{\varepsilon}{3}$ holds for all $\delta \geq N_3$. Therefore, one derives that
\begin{equation}
\begin{aligned}
\chi_3(\delta)&=\parallel \mathscr{P}_2(\kappa)\mathscr{P}_3(\eta) \parallel \\
&\leq \parallel \mathscr{P}_2(\kappa)\parallel \times \parallel\mathscr{P}_3(\eta) \parallel \leq \parallel\mathscr{P}_3(\eta)\parallel\leq\frac{\varepsilon}{3}.
\end{aligned}
\end{equation}

To sum up, for any $\varepsilon>0$, we can define an integer $N=N_1+\max\{N_2,N_3\}$ such that
\begin{equation}
\begin{aligned}
&\| \mathscr{P}_2(\delta)-{\bf 1}^\top_{s}\otimes \vec{\omega} \|\\
&\leq \|\mathscr{P}_1(N_1)\mathscr{P}_2(\eta-N_1)-({\bf 1}^\top_{s}\otimes \vec{\omega}) \mathscr{P}_2(\eta-N_1)\| \\
&~ + \parallel({\bf 1}^\top_{s}\otimes \vec{\omega}) \mathscr{P}_2(\eta-N_1) - {\bf 1}^\top_{s}\otimes \vec{\omega} \parallel + \parallel\mathscr{P}_2(N_1)\mathscr{P}_3(\eta-N_1)\|\\
&\leq \frac{\varepsilon}{3}+\frac{\varepsilon}{3}+\frac{\varepsilon}{3}=\varepsilon
\end{aligned}
\end{equation}
holds for any $\delta\geq N$.

Accordingly, one can conclude that $\lim\limits_{\delta\rightarrow+\infty}\| \mathscr{P}_2(\delta)-{\bf 1}^\top_{s}\otimes \vec{\omega} \|=0$; it implies that $\lim\limits_{\delta\rightarrow+\infty}\| \mathscr{P}_2(\delta)\parallel={\bf 1}^\top_{s}\otimes \vec{\omega}$. Therefore, one has that
\begin{equation}
\begin{aligned}
\lim_{\delta\rightarrow+\infty}&\mathscr{P}(\delta)=\lim_{\delta\rightarrow+\infty}\left( \begin{array}{cc} \mathscr{P}_1(\delta) & \mathscr{P}_2(\delta) \\ {\bf 0}_{(2^n-s) \times s} & \mathscr{P}_3(\delta) \end{array}\right)\\
&=\left( \begin{array}{cc} {\bf 1}^\top_{s}\otimes \vec{\omega} & {\bf 1}^\top_{s}\otimes \vec{\omega} \\ {\bf 0}_{(2^n-s) \times s} & {\bf 0}_{(2^n-s) \times s} \end{array}\right)={\bf 1}_s^\top \otimes \omega.
\end{aligned}
\end{equation}
It indicates that this \texttt{PBN} (\ref{equ-pbn-assr}) is globally \texttt{SPD} w.r.t. $\omega\in \nabla_{2^n}$, and the verification of Theorem \ref{thm-equivalence} is established.
\end{proof}

\section*{Acknowledgments}
The authors are sincerely grateful to Dr. Eyal Weiss for his constructive suggestions to the previous version of this paper.



\end{document}